\documentclass[%
 reprint,
superscriptaddress,
preprintnumbers,
nofootinbib,
 amsmath,amssymb,
 aps,
 prb,
]{revtex4-2}

\usepackage{graphicx}
\usepackage{dcolumn}
\usepackage{bm}
\usepackage{hyperref}
\usepackage{color}

\hypersetup{
    unicode=false,          
    pdftoolbar=true,        
    pdfmenubar=true,        
    pdffitwindow=false,     
    pdfstartview={FitH},    
    pdftitle={},    
    pdfauthor={},     
    pdfsubject={},   
    pdfcreator={},   
    pdfproducer={}, 
    pdfkeywords={} {} {}, 
    pdfnewwindow=true,      
    colorlinks=true,       
    linkcolor=blue, 
    citecolor=blue,        
    filecolor=magenta,      
    urlcolor=blue,           
	breaklinks=true
} 

\usepackage{amssymb}
\usepackage{amsmath}
\usepackage{amsfonts}
\usepackage{xspace}
\usepackage{bm,bbm}
\usepackage{relsize}
\usepackage{siunitx}
\usepackage{xcolor}
\usepackage{upgreek}

\makeatletter
\DeclareFontFamily{OMX}{MnSymbolE}{}
\DeclareSymbolFont{MnLargeSymbols}{OMX}{MnSymbolE}{m}{n}
\SetSymbolFont{MnLargeSymbols}{bold}{OMX}{MnSymbolE}{b}{n}
\DeclareFontShape{OMX}{MnSymbolE}{m}{n}{
    <-6>  MnSymbolE5
   <6-7>  MnSymbolE6
   <7-8>  MnSymbolE7
   <8-9>  MnSymbolE8
   <9-10> MnSymbolE9
  <10-12> MnSymbolE10
  <12->   MnSymbolE12
}{}
\DeclareFontShape{OMX}{MnSymbolE}{b}{n}{
    <-6>  MnSymbolE-Bold5
   <6-7>  MnSymbolE-Bold6
   <7-8>  MnSymbolE-Bold7
   <8-9>  MnSymbolE-Bold8
   <9-10> MnSymbolE-Bold9
  <10-12> MnSymbolE-Bold10
  <12->   MnSymbolE-Bold12
}{}

\let\llangle\@undefined
\let\rrangle\@undefined
\DeclareMathDelimiter{\llangle}{\mathopen}%
                     {MnLargeSymbols}{'164}{MnLargeSymbols}{'164}
\DeclareMathDelimiter{\rrangle}{\mathclose}%
                     {MnLargeSymbols}{'171}{MnLargeSymbols}{'171}
\makeatother

\renewcommand{\vec}[1]{\bm{#1}}

\let\originalleft\left
\let\originalright\right
\renewcommand{\left}{\mathopen{}\mathclose\bgroup\originalleft}
\renewcommand{\right}{\aftergroup\egroup\originalright}

\usepackage{amsthm}

\theoremstyle{plain}
\newtheorem{thm}{\protect\theoremname}[section]
\theoremstyle{definition}
\newtheorem{example}[thm]{\protect\examplename}
\theoremstyle{plain}
\newtheorem{lem}[thm]{\protect\lemmaname}
\theoremstyle{plain}
\newtheorem{prop}[thm]{\protect\propositionname}
\theoremstyle{remark}
\newtheorem{rem}[thm]{\protect\remarkname}
\theoremstyle{definition}
\newtheorem{defn}[thm]{\protect\definitionname}
\theoremstyle{plain}
\newtheorem{conj}[thm]{\protect\conjecturename}

\makeatother

\providecommand{\examplename}{Example}
\providecommand{\lemmaname}{Lemma}
\providecommand{\propositionname}{Proposition}
\providecommand{\theoremname}{Theorem}
\providecommand{\conjecturename}{Conjecture}
\providecommand{\remarkname}{Remark}
\providecommand{\definitionname}{Definition}

\begin{document}
\title{Classification of Fragile Topology Enabled by Matrix Homotopy}

\author{Ki Young Lee}
\email[Email: ]{kiylee@sandia.gov}
\affiliation{Center for Integrated Nanotechnologies, Sandia National Laboratories, Albuquerque, New Mexico 87185, USA}

\author{Stephan Wong}
\affiliation{Center for Integrated Nanotechnologies, Sandia National Laboratories, Albuquerque, New Mexico 87185, USA}
 
\author{Sachin Vaidya}
\affiliation{Department of Physics, Massachusetts Institute of Technology, Cambridge, Massachusetts 02139, USA}

\author{Terry A. Loring}
\affiliation{Department of Mathematics and Statistics, University of New Mexico, Albuquerque, New Mexico, 87131, USA}

\author{Alexander Cerjan}
\email[Email: ]{awcerja@sandia.gov}
\affiliation{Center for Integrated Nanotechnologies, Sandia National Laboratories, Albuquerque, New Mexico 87185, USA}

\date{\today} 

\begin{abstract}

The moir\'{e} flat bands in twisted bilayer graphene have attracted considerable attention not only because of the emergence of correlated phases but also due to their nontrivial topology.
Specifically, they exhibit a new class of topology that can be nullified by the addition of trivial bands, termed fragile topology, which suggests the need for an expansion of existing classification schemes.
Here, we develop a $\mathbb{Z}_2$ energy-resolved topological marker for classifying fragile phases using a system's position-space description, enabling the direct classification of finite, disordered, and aperiodic materials.
By translating the physical symmetries protecting the system's fragile topological phase into matrix symmetries of the system's Hamiltonian and position operators, we use matrix homotopy to construct our topological marker while simultaneously yielding a quantitative measure of topological robustness.
We show our framework's effectiveness using a $C_2\mathcal{T}$-symmetric twisted bilayer graphene model and photonic crystal as a continuum example.
We have found that fragile topology can persist both under strong disorder and in heterostructures lacking a bulk spectral gap, and even an example of disorder-induced re-entrant topology.
Overall, the proposed scheme serves as an effective tool for elucidating aspects of fragile topology, offering guidance for potential applications across a variety of experimental platforms from topological photonics to correlated phases in materials.
\end{abstract}

\maketitle

Since their discovery, robust localized states have symbolized topological phases of matter and served as pivotal physical quantities that broaden the frontiers of functional materials. Representative examples include Chern insulators~\cite{Haldane1988} and topological insulators~\cite{Kane2005}, which host gapless edge states protected by associated topological invariants and persisting as long as their internal symmetries are preserved and the corresponding spectral gap remains open. Such classes of topology are now referred to as strong and exhibit a bulk-boundary correspondence, linking the presence of protected boundary states to the nonzero topological indices of the bulk bands. Within a $K$-theoretic framework~\cite{Kitaev2009, Schnyder2008, Ryu2010}, topological classification schemes rooted in vector bundles can systematically assign topological indices to band gaps by summing up the topological invariants of each occupied band. Naturally, these strong topological phases are also stable in the sense that the addition of a trivial occupied band does not change their behavior, as the additional trivial topological invariant does not alter the overall indices of the system. 

Recently, the emergence of crystalline-symmetry protected topological phases~\cite{Fu2011,Kruthoff2017,Benalcazar2017} has revealed a richer variety of phenomena such as weaker forms of topology. Distinct from strong topology, crystalline topology often lacks a general bulk-boundary correspondence~\cite{Zhu2023,Chen2023}, stimulating the development of new theoretical frameworks such as topological quantum chemistry~\cite{Bradlyn2017}, symmetry indicators~\cite{Po2017, benalcazar2019quantization}, and Wilson loop approaches~\cite{Bouhon2019,Wang2019} for classifying these subtle phases. Among these weak topological phases, the concept of fragile topology~\cite{Po2018} has been found as a contrasting case to stable topology, where the topological features of certain bands can be trivialized through the addition of trivial atomic bands. Importantly, the moir{\' e} flat bands of small-angle twisted bilayer graphene (TBG)~\cite{Bistritzer2011,Ahn2019,Song2019,Hwang2019,Po2019,Bradlyn2019,Zou2018,Calugaru2022} have been identified as possessing such fragile topology, attracting significant attention due to potential connections to correlated  phases~\cite{Cao2018,Else2019,Liu2019,Turner2022} and superconductivity~\cite{Cao2018a,Peotta2015,Xie2020,Po2018a,Peri2021}. This new class of topology also appears relevant to debates regarding the weak topological protection~\cite{Rosiek2023,Xu2023,Ding2024,Palmer2021} of the interface states observed in various photonic crystal platforms, as well as to potential applications of corner states and some types of edge states~\cite{Peri2020,Wu2024,Manes2020,Li2020,Vaidya2023,Wang2024,Ghorashi2024,DePaz2019,schulz2021topological}. However, most studies of fragile topology have focused on momentum-space classifications~\cite{Po2018,Ahn2019,Hwang2019,song2020,Song2020a,Lian2020,Lange2023}, which have exhibited limitations in both developing invariants that identify fragile phases, as well as quantitatively analyzing experimentally relevant scenarios such as finite-size effects, disorder, and heterostructures. Moreover, the ability to classify aperiodic materials is crucial for the continued exploration of twisted materials.


Here, we derive and demonstrate a position-space approach for classifying fragile topology and find an associated measure of robustness.  
Specifically, we show how the physical symmetries that protect fragile topology can give rise to matrix symmetries in a system's Hamiltonian and position operators when expressed in an atypical basis. These matrix symmetries can then be used to define a homotopic invariant that distinguishes 2D systems based on which atomic limit they can be continued to, yielding an energy-resolved $\mathbb{Z}_2$ topological marker.
Moreover, this approach implicitly introduces a quantitative measure of topological protection that remains valid under a variety of conditions, including finite system size, disorder, and environmental perturbations. We apply our approach to a disordered $C_2 \mathcal{T}$-symmetric TBG model and a two-dimensional (2D) photonic crystal embedded in an air background, showing that fragile topology can persist even under strong disorder and when a heterostructure lacks a bulk spectral gap. Our disordered simulations also reveal a disorder-induced re-entrant transition~\cite{Rieder2013,Zhang2025} to a fragile phase with increasing disorder strength.
Altogether, our work provides guidance for identifying and characterizing fragile topology beyond momentum-space descriptions, potentially offering insights into its relevance in aperiodic and moir\'{e} systems, photonic or metamaterial applications, correlated physics, and superconductivity.

We start by deriving an energy-resolved $\mathbb{Z}_2$ marker that classifies fragile topology in finite $C_2\mathcal{T}$-symmetric systems.
The key difference of our approach is that rather than focusing on a system's eigenstates, as are used in standard classification methods for fragile topology such as Wilson loops \cite{cano2018topology, Bouhon2019, Hwang2019, Bradlyn2019, de2019tutorial, DePaz2019, Vaidya2023}, our framework is instead rooted in the system's operators directly, such that these eigenstates never need to be determined.
By definition, the Hamiltonian of a $C_2\mathcal{T}$-symmetric 2D system obeys
\begin{equation}
    \left(C_2 \mathcal{T}\right)^{-1} H \left(C_2 \mathcal{T}\right) = H,
\end{equation}
while the system's position operators $X$ and $Y$ anti-commute with this symmetry due to $C_2$,
\begin{align}
    \left(C_2 \mathcal{T}\right)^{-1} X \left(C_2 \mathcal{T}\right) & = -X, \
    \left(C_2 \mathcal{T}\right)^{-1} Y \left(C_2 \mathcal{T}\right) & = -Y.
\end{align}
This suggests that we can define a transpose-like matrix operation $\rho$ as
\begin{equation}
    M^\rho = \left(C_2 \mathcal{T}\right)^{-1} M^\dagger \left(C_2 \mathcal{T}\right),
\end{equation}
which can be simplified using standard techniques (see Supplementary Sec.~SI.B) to
\begin{equation}
    M^\rho = C_2 M^\top C_2. \label{eq:rhoDef}
\end{equation}
In other words, this new operation is the transpose intertwined with a rotation, and associates the system's physical $C_2 \mathcal{T}$-symmetry to a set of mathematical matrix symmetries for its operators as
\begin{equation}
   X^\rho = -X, \quad    
    Y^\rho = -Y  \enspace\mbox{and}\enspace  
      H^\rho = H . \label{eq:XYHrho}
\end{equation}

The relations in Eq.~\eqref{eq:XYHrho} are reminiscent of the system's operators being symmetric $M^\top = M$, or skew-symmetric $M^\top = -M$, which, in general, are useful properties for studying matrix homotopy. For example, two invertible Hermitian skew-symmetric matrices $M_0$ and $M_1$ are homotopy equivalent and can be connected via a path of invertible Hermitian skew-symmetric matrices $M_t$ with $t \in [0,1]$ if and only if $\textrm{sign}[\textrm{Pf}(M_0)] = \textrm{sign}[\textrm{Pf}(M_1)]$, where $\textrm{Pf}$ denotes the Pfaffian. This is because the Pfaffian can only change sign when two of eigenvalues of some $M_t$ simultaneously cross $0$ (see Remark SI.6 and \cite[\S 3.9]{Serre2010Matrices}). To take advantage of these prior results on matrix homotopy, we use the unitary matrix
\begin{equation}
    W = \frac{1}{\sqrt{2}}\left(C_2+i\mathbf{1}\right)
\end{equation}
to find a basis in which $\rho$ is recast as $\top$, such that
\begin{equation}
    W M^\rho W^\dagger = \left(W M W^\dagger \right)^\top,
\end{equation}
see Supplementary Sec.~SI.C. Specifically, this means that $W X W^\dagger$ and $W Y W^\dagger$ are Hermitian skew-symmetric, while $W H W^\dagger$ is Hermitian symmetric.

To develop an invariant to classify fragile topology protected by $C_2 \mathcal{T}$-symmetry, we combine the system's operators centered about a choice of $(x,y,E)$ in position-energy space using the Pauli matrices $\sigma_{x,y,z}$
\begin{align}
\label{eq:loc}
\begin{split}
L&_{(x,y,E)}(WXW^\dagger,WYW^\dagger,WHW^\dagger) \\
&= \kappa(WXW^\dagger-x\mathbf{1})\otimes\sigma_x+ \kappa(WYW^\dagger-y\mathbf{1})\otimes\sigma_z \\
&\quad + (WHW^\dagger-E\mathbf{1})\otimes\sigma_y,  
\end{split}
\end{align}
yielding a spectral localizer \cite{Loring2015,Loring2017,Loring2020,Cerjan2024a}. Here, $\kappa$ is a scaling coefficient that sets the spectral weight of the position operators relative to the Hamiltonian. In spectrally gapped systems, $\kappa$ is typically on the order of $\kappa \sim E_{\textrm{gap}} / L_{\textrm{min}}$, where $E_{\textrm{gap}}$ is the width of the gap and $L_{\textrm{min}}$ is the minimum length of the system in any direction \cite{Loring2020}. Heuristically, the use of the Pauli matrices is preserving the independence of the information carried in $X$, $Y$, and $H$, in an analogous manner to how the Pauli matrices (along with $\mathbf{1}$) form a complete basis for $2$-by-$2$ Hermitian matrices.

\begin{figure}[t]
\center
\includegraphics[width=\columnwidth]{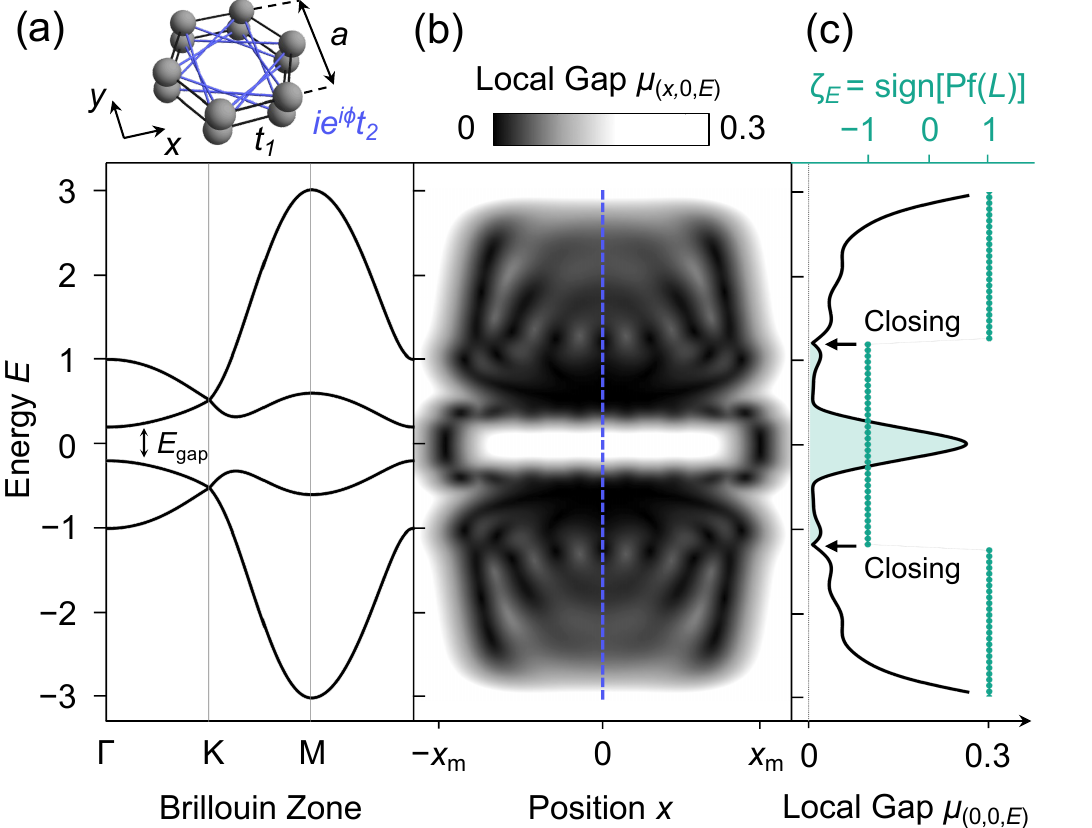}
\caption{ 
(a) Bulk band structure for the $C_2\mathcal{T}$-symmetric TBG model. Inset diagrams this lattice model. 
(b) The local gap $\mu_{(x,y,E)}$ for $x$ and $E$ at $y = 0$. Position $x$ is scaled in terms of the lattice constant $a$ and $x_m$ denotes the length of the edge from the origin. 
(c) $\mu_{(x,y,E)}$ for $E$ at the center of rotation denoted by blue dotted line in panel (b). The energy-resolved invariant $\zeta_E$ is shown by the green dots, where $\kappa = 0.1 t/a$ for all calculations.
}
\label{fig:fig1}
\end{figure}

The key features of Eq.~\eqref{eq:loc} are that it is Hermitian for any $(x,y,E)$, and at $x=y=0$ (i.e., the center of rotation), $L_{(0,0,E)}$ is skew-symmetric. This skew-symmetry has been achieved by placing $W X W^\dagger$ and $W Y W^\dagger$ against $\sigma_x$ and $\sigma_z$, and placing $W H W^\dagger$ against $\sigma_y$, such that each tensor product is between a symmetric matrix and a skew-symmetric matrix. Thus, the energy-resolved invariant
\begin{multline}
\label{eq:inv}
\zeta_E(X,Y,H) \\
= \textrm{sign}[\textrm{Pf}(L_{(0,0,E)}(WXW^\dagger,WYW^\dagger,WHW^\dagger))]
\end{multline}
distinguishes systems that can be connected to each other while preserving $C_2 \mathcal{T}$-symmetry and maintaining a positive local gap
\begin{multline}
\label{eq:mu}
\mu_{(x,y,E)}(X,Y,H) = \text{min}(|\text{spec}[L_{(x,y,E)}(X,Y,H)]|), 
\end{multline}
where $\textrm{spec}[M]$ is the spectrum of $M$.
By definition, $\zeta_E\in \{-1,+1\} \cong \mathbb{Z}_2$; if $\zeta_E= -1$, the system at $E$ cannot be connected to a trivial atomic limit, and vise versa. Moreover, since $\zeta_E$ cannot change its value without $\mu_{(0,0,E)} \rightarrow 0$, either by changing the choice of $E$ or by perturbing the system, the local gap serves as a quantitative measure of the system's topological protection at $E$. Specifically, when perturbing the system $H \rightarrow H + \delta H$, $\zeta_E$ is guaranteed by Weyl's inequality to be preserved so long as $\Vert \delta H \Vert < \mu_{(0,0,E)}(X,Y,H)$ \cite{weyl_asymptotische_1912,Bhatia1997}. In addition, as locations where $\mu_{(x,y,E)}=0$ are associated with the locations of a system's states \cite{cerjan_V_L2023quadraticPS}, changes in a system's topological marker necessarily imply changes in the structure of its states. A detailed mathematical discussion of Eq.~\eqref{eq:inv}, its essential properties, and its relation to atomic limits is given in Supplementary Secs.~SI.E-G. In particular, Examples SI.11 and SI.12 show the form of the two different classes of atomic limits distinguished by $\zeta_E$.

Summarizing our derivation, we first used a system's physical symmetries to define a transpose-like matrix operation that translates the physical symmetries to matrix symmetries of the system's $H$, $X$, and $Y$. Then, we found a change of basis that transformed this matrix operation into the standard matrix transpose, such that in this atypical basis the system's operators were either symmetric or skew-symmetric. Finally, by tensoring these operators using the Pauli matrices, we formed a single skew-symmetric matrix whose Pfaffian's sign discriminates between which atomic limits a given system can be connected to without closing the system's local gap. Altogether, by using results from matrix homotopy, this argument yields an energy-resolved invariant for classifying fragile topology as well as a quantitative measure of topological protection.

Having derived a classification framework applicable to finite systems, we demonstrate its use in a four-band model that is a low-energy approximation of TBG and exhibits fragile topology~\cite{Ahn2019,Po2019}. This model consists of a bilayer honeycomb lattice, as schematically shown in Fig.~\ref{fig:fig1}(a), where $t_1$ and $t_2$ represent the intra- and inter-layer hopping amplitudes, respectively. The blue lines spirally connecting inter-layer sites represent next-nearest neighbor (NNN) hoppings with the hopping phase $\pm \phi$, which can induce a nontrivial fragile band gap. We have provided the expressions of the Hamiltonian in both position and momentum space the End Matter and further information in Supplementary Sec.~SII. 

Comparison of the bulk band structure of the infinite four-band TBG model with the local gap of a finite system confirms that the locations in $(x,y,E)$-space with $\mu_{(x,y,E)} \approx 0$ indicate the presence of states at the specified energies and positions, see Figs.~\ref{fig:fig1}(a),(b). For choices of $E$ residing within the spectral extent of the bulk bands, extended Bloch states are distributed throughout the system, whereas within the bulk band gap, only localized states exist at the system's boundaries. Note that the fluctuation pattern of $\mu_{(x,0,E)}$ only intermittently touching zero near the band gap inherently suggests the weak topological nature of this fragile system; strong topological phases instead exhibit a spheroid of appropriate dimension where $\mu_{(\mathbf{x},E)} = 0$ (see Supplementary Sec.~SI.G).

Within the bulk band gap, the energy-resolved marker $\zeta_E$ confirms the finite system's fragile topology, while the large local gap at the rotation center indicates this phase's strong topological protection, see Fig.~\ref{fig:fig1}(c). For energies above and below the bulk gap, $\zeta_E$ maintains a nontrivial value of $-1$ until the first closing points of $\mu_{(0,0,E)}$, beyond which it switches to $+1$. Although the exact energy where $\mu_{(0,0,E)} =0$ may vary with the parameter $\kappa$ (see Supplementary Sec.~SIII), these spectral regions with small local gaps are not topologically robust, as very small system perturbations are able to change the topology. 

To confirm that the local fragile marker $\zeta_E$ captures phase transitions, we uniformly vary $\phi$ between all NNN sites from $-\pi$ to $\pi$. As can be seen in Fig.~\ref{fig:fig2}(a), the width of the bulk spectral gap under this alteration is symmetric about $\phi=0$ and touches zero twice at $\phi=\pi /3$ and $\phi=2\pi /3$. Similarly, the local gap closes at precisely the same points where $E_{\textrm{gap}}=0$ and $\zeta_E$ changes across these values of $\phi$, indicating a change in the material's fragile topological phase.

\begin{figure}[t]
\center
\includegraphics[width=\columnwidth]{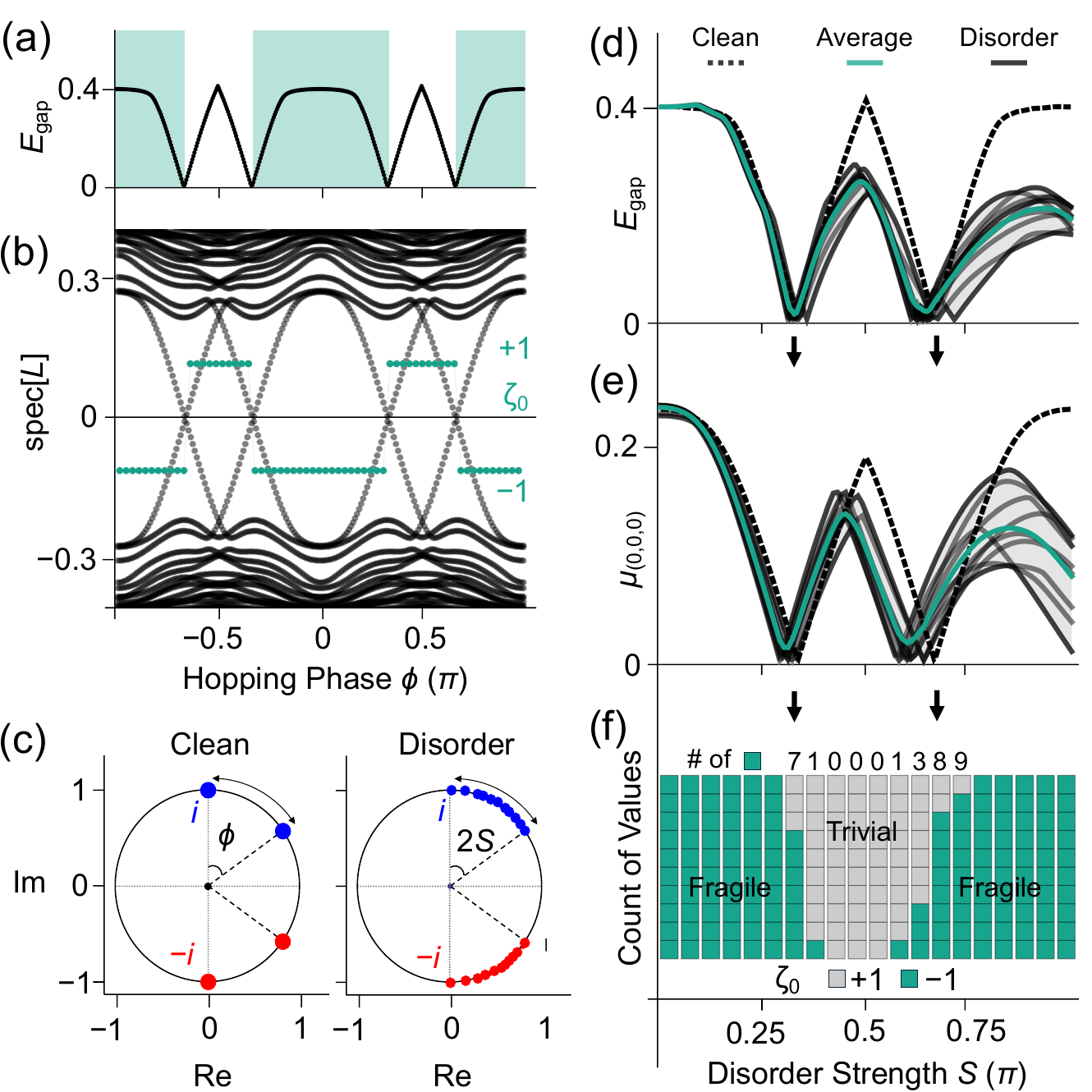}
\caption{
(a) Closings of $E_{\textrm{gap}}$ due to uniform changes in $\phi$ in the periodic $C_2\mathcal{T}$-symmetric TBG model. 
(b) Spectrum of $L_{(0,0,0)}$ (gray) and $\zeta_0$ (green) at the middle of the bulk band gap as $\phi$ is uniformly varied in the finite TBG system.
(c) Hopping phases of clean and disordered systems depicted on the unit circle. The red and blue dots correspond to the opposite off-diagonal terms of the Hamiltonian, forming conjugate pairs. 
(d), (e), (f) Ensemble analysis of $E_{\textrm{gap}}$ (d), $\mu_{(0,0,0)}$ (e), and histogram of $\zeta_{0}$ for the 10 disorder configurations for increasing $S$ (f). In (d) and (e), the black dashed lines show the behavior for uniform hopping phase changes in the clean system. Solid Gray lines show the results for each disorder configuration, the solid green line shows the average over the ensemble, and the gray shading fills the area between the maximum and minimum of these data, representing the sample deviations. The histogram in (f) uses green and gray to indicate $\zeta_0=-1$ and $+1$, respectively.
}
\label{fig:fig2}
\end{figure}

\begin{figure*}[t]
\center
\includegraphics[width=2\columnwidth]{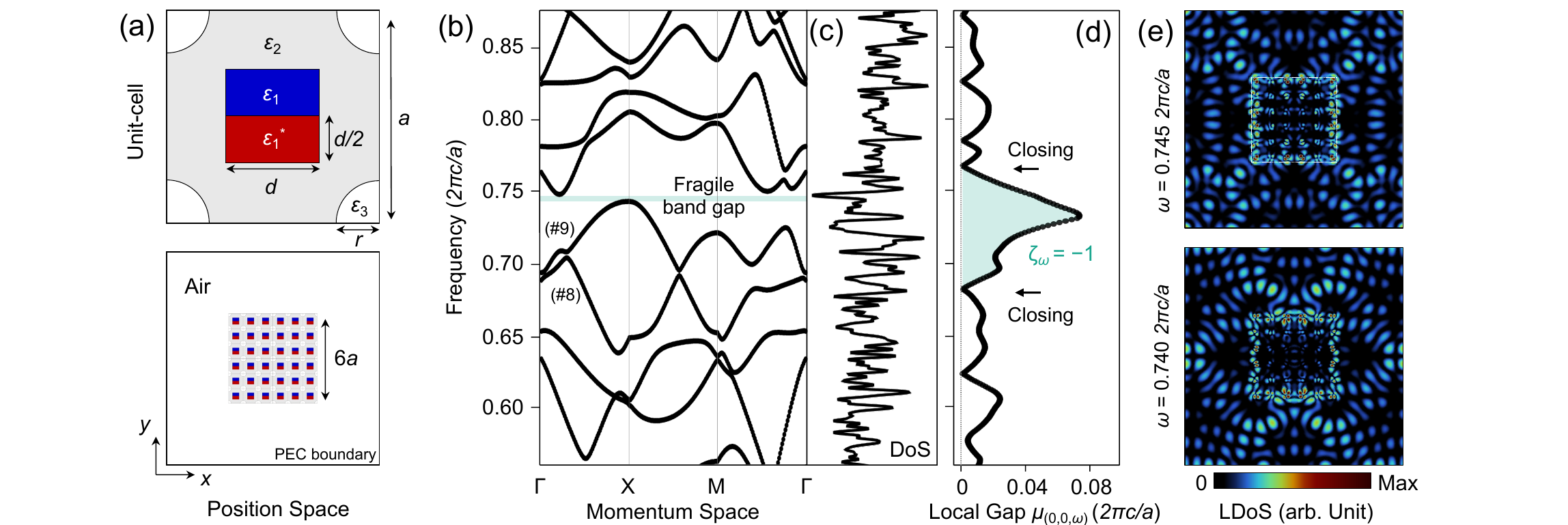}
\caption{(a) Schematics of a single unit cell of the photonic crystal (upper) and a heterostructure (lower) formed by $6\times6$ unit cells surrounded by air bounded by a perfect electrical conductor (PEC). Here, $\epsilon_1$ =\{16, 6$i$, 0; $-$6$i$, 16, 0; 0, 0, 16\}, $\epsilon_2=4$, $\epsilon_3=1$, $a$ is the lattice constant, $r = 0.2a$, and $d = 0.45a$.  
(b) Band structure of the bulk photonic crystal.
(c),(d) The density of states (DoS) (c) and the local gap and frequency-resolved index (d) for the $6\times6$ finite system from (a). Calculations in (d) use $\kappa$ = 0.01 $2\pi c/a^2$, where $c$ is the speed of light in a vacuum.
(e) Local density of states (LDoS) within the fragile band gap depicted by green in (b) at the two frequencies for the finite heterostructure shown in (a).
}
\label{fig:fig3}
\end{figure*}

As our framework works directly with a finite system expressed in position-space, it can inherently be applied to disordered and aperiodic systems. To illustrate this capability, we consider an ensemble of disordered variants of the four-band TBG model where the hopping phases $\phi_{jk}$ between each pair of NNN sites $j$ and $k$ are randomly assigned a value within an angle range of 2$S$ from $\pm i$ from a uniform distribution while preserving $C_2\mathcal{T}$-symmetry. Therefore, $S$ represents the median value of the disorder strength, allowing us to investigate the phase diagram of the disordered system based on this variable. In Figs.~\ref{fig:fig2}(d), (e), and (f), we present $E_{\textrm{gap}}$, $\mu_{(0,0,E)}$, and $\zeta_E$ as functions of $S$ for 10 different disorder realizations. We numerically observe that the disordered samples exhibit spectral gap closings and topological phase transitions near $S=\pi/3$ and $2\pi/3$, confirming our energy-resolved local marker's ability to classify disordered systems.

Moreover, the results in Figs.~\ref{fig:fig2}(d), (e), and (f) provide clear evidence of disorder-induced re-entrant topological phase transitions~\cite{Rieder2013, Zhang2025} in fragile topology, offering a novel perspective on the stability of fragile topological phases. These findings reveal that obstructions to connecting a system to a trivial atomic limit can reemerge beyond a certain disorder threshold rather than simply being destroyed. In particular, such behavior parallels phenomena observed in topological Anderson insulators~\cite{Li2009}, suggesting an intricate interplay between disorder and topology in moir{\' e} systems. Thus, the re-entrant topological transition emphasizes the limitations of momentum-space approaches and illustrates the possibilities for a topological marker rooted in a system's position-space description.

Finally, to show the broad applicability of our approach in quantifying material topology from tight-binding to continuum models, traditionally challenging due to high computational costs~\cite{Bianco2011,li2024}, we turn to classifying fragile topology in photonic crystals (PhC) surrounded by air~\cite{Dixon2023}. We consider the 2D PhC unit cell structure depicted in Fig.~\ref{fig:fig3}(a) that is designed to be $C_2\mathcal{T}$-symmetric while not exhibiting either symmetry in isolation. The full heterostructure consists of a region containing $6\times6$ unit cells surrounded by air. The system is discretized with standard finite-difference methods to define the Hamiltonian and position operators~\cite{Cerjan2022a,Dixon2023,Cerjan2024,Cerjan2024a} and we focus on its transverse electric (TE) modes with non-zero electromagnetic field components $(H_z, E_x, E_y)$.



For the infinite PhC system, the bulk band structure [Fig.~\ref{fig:fig3}(b)] possesses two bands that can be shown to exhibit fragile topology using Wilson loops (See Supplementary Sec.~SIV). However, the inclusion of the surrounding air region in the finite heterostructure removes any spectral gap from the system's DoS, as shown in Fig.~\ref{fig:fig3}(c). Nevertheless, the frequency-resolved marker $\zeta_\omega$ identifies this frequency range as possessing fragile topology despite the lack of a spectral gap, see Fig.~\ref{fig:fig3}(d). Moreover, although the bulk band gap above the fragile bands is small, the relatively large local gap indicates that the system's topology is more robust than would be suggested by the narrow width of the bulk band gap. The LDoS for the magnetic field intensity near the frequencies of the fragile band gap are shown in Fig.~\ref{fig:fig3}(e). At both frequencies inside the fragile band gap, the LDoS is localized at the boundaries of the PhC and has substantial support outside of the PhC. Therefore, these results confirm that our framework for classifying fragile topology can be applied to systems that lack a bulk spectral gap without alteration and that these systems can still exhibit topological robustness associated with the region responsible for the fragile topology.

In conclusion, we have introduced a position-space framework for classifying fragile topology rooted in matrix homotopy that distinguishes systems based on which atomic limits they can be continued to. Applying this framework to the TBG lattice model and PhC continuum model, we have shown the breadth of the local marker's ability to capture systems' fragile nature. Moreover, the versatility of our method is highlighted by its ability to identify nontrivial fragile phases under conditions of strong disorder and gapless environments. 
Taking disorder-induced re-entrant fragile topology, which is unpredictable by conventional methods, as an example, our approach highlights further opportunities for potential applications across various fields such as correlated physics, metamaterials, and topological photonics.

\section*{Acknowledgments}
K.Y.L.\ and S.W.\ acknowledge support from the Laboratory Directed Research and Development program at Sandia National Laboratories.
S.V.\ acknowledges support from the U.S.\ Office of Naval Research (ONR) Multidisciplinary University Research Initiative (MURI) under Grant No.\ N00014-20-1-2325 on Robust Photonic Materials with Higher-Order Topological Protection.
T.A.L.\ acknowledges support from the National Science Foundation, Grant No. DMS-2349959.
A.C.\ acknowledges support from the U.S.\ Department of Energy, Office of Basic Energy Sciences, Division of Materials Sciences and Engineering.
This work was performed in part at the Center for Integrated Nanotechnologies, an Office of Science User Facility operated for the U.S. Department of Energy (DOE) Office of Science.
Sandia National Laboratories is a multimission laboratory managed and operated by National Technology \& Engineering Solutions of Sandia, LLC, a wholly owned subsidiary of Honeywell International, Inc., for the U.S. DOE's National Nuclear Security Administration under Contract No. DE-NA-0003525. 
The views expressed in the article do not necessarily represent the views of the U.S. DOE or the United States Government.

\section*{End Matter}

Here, we provide detailed expressions of a four-band tight-binding model known to describe the essential physics of the nearly flat bands of TBG according to a low-energy continuum theory. As depicted in Fig.~\ref{fig:fig1}(a), the model consists of a honeycomb lattice with two layers. The explicit expression of the Hamiltonian for this lattice model in momentum space is given by
\begin{multline}
\label{eq:mo_H}
h(\vec{k}) = 
\hat{t}_1 \otimes \left[ \left( 1 + 2 \cos \frac{\sqrt{3} k_x a}{2} \cos \frac{k_y a}{2} \right) \sigma_x \right.  \\ 
\left. \quad + 2 \sin \frac{\sqrt{3} k_x a}{2} \cos \frac{k_y a}{2} \sigma_y \right]  
+ \hat{t}_2 \otimes \mathbf{1} \left[ n(\vec{k}, \phi)+n^*(\vec{k}, \phi) \right],       
\end{multline}
\noindent
with
\begin{align}
\label{eq:NNN}
\begin{split}
n(\vec{k}, \phi) = i e^{i\phi} \left(e^{i\mathbf{k}\cdot\mathbf{a}_1}+e^{-i\mathbf{k}\cdot\mathbf{a}_2}+e^{i\mathbf{k}\cdot\mathbf{a}_3}   \right) ,    
\end{split}
\end{align}
where we employ $\hat{t}_1 = 0.4 t\mathbf{1} + 0.6  t \tau_z$ and $\hat{t}_2 =  0.1  t \tau_x$, meaning intra- and inter-layer hopping amplitudes, respectively, as schematically shown by the black and green lines in Fig.~\ref{fig:fig1}(a). Here, $\mathit{t}$ indicates the overall energy scale, $\mathbf{k} = (k_x, k_y)$ is the in-plane momentum, and the primitive vectors are $\mathbf{a}_{1,2} = (\sqrt{3}, \pm 1)a/2$ and $\mathbf{a}_3=\mathbf{a}_1-\mathbf{a}_2$. The Pauli matrices $\sigma_{x,y,z}$ and $\tau_{x,y,z}$ represent the sublattice and orbital degrees of freedom, respectively. The blue lines, spirally connecting inter-layer sites, denote the next nearest neighbor hoppings that are crucial for inducing a nontrivial fragile band gap. We consider the additional inter-layer hopping phase to be $\phi = 0$ here so that the initial coefficient of inter-layer hopping is purely imaginary. 

In Fig.~\ref{fig:fig1}(a), we show the bulk band structure of the four-band model. All four energy bands are symmetric with respect to the Fermi level and each pair of bands exhibits Dirac points at each of the $K$ and $K'$ points throughout the Brillouin zone. These Dirac points are protected by space-time $C_2\mathcal{T}$ inversion symmetry~\cite{Ahn2019}, where $\mathcal{T}^2 = \mathbf{1}$ and $C_2$ denotes a twofold rotation about the $z$-axis.
In addition, the valence bands exhibit an obstruction that prohibits their representation by exponentially localized Wannier functions that obey the system's $C_2\mathcal{T}$-symmetry. As this obstruction disappears when more trivial bands are added to the model, this system exhibits fragile topology~\cite{Ahn2019,Po2019}. 

To study finite geometries and disorder in the four-band TBG model, we instead work with the system's expression in position-space, which is given by the Hamiltonian
\begin{equation}
\label{eq:po_H}
H = \sum_{\langle i,j \rangle} c_i^\dagger (\hat{t}_1)_{ij}c_j + \sum_{\langle\langle i,j \rangle\rangle} c_i^\dagger s_{ij}( ie^{i \phi}\hat{t}_2)_{ij} c_j, 
\end{equation}
where $s_{ij}=+1$ is chosen for $\mathbf{r}_i=\mathbf{r}_j+a\mathbf{y}$.

\section*{Supplementary Information}
The Supplementary Information provides a full treatment of the mathematics that underpin the results discussed in the main text, a discussion of the edge states that appear in the TBG model with open boundaries, an elaboration on the spectral weight factor $\kappa$ used in the spectral localizer framework, Wilson loop calculations for the 2D fragile photonic crystal, and references \cite{loringSor2013AC_time_reversal,loringSor2014AC_real_orthogonal,loringSor2016AC_real_symmetric,boersemaL2015_K-theory_by_symmetries,Gantmacher1959matrix_theory,Lax2007Matrices,herrera2024SymmetryBootstrap,EndersShulman2023AC_matrices_and_dimension,bratteli2012OperAlg_StatMec,cerjanLoring2024even_spheres,Boersema2020_K-theory_by_symmetriesII,loring2014quantitative}.



\renewcommand{\thesection}{S\Roman{section}}
\setcounter{section}{0}

\renewcommand{\thefigure}{S\arabic{figure}}
\setcounter{figure}{0}

\renewcommand{\theequation}{S\arabic{equation}}
\setcounter{equation}{0}

\section{Deriving the invariant for classifying fragile topology \label{sec:S1}}

The main text provides a derivation of a topological marker for classifying fragile topology using results from matrix homotopy. However, underpinning this entire derivation is a real $C^*$-algebra and associated techniques. Thus, here, as we provide further details on the argument given in the main text, we will also place these arguments into the context of the study of real $C^*$-algebras.

Broadly, a (complex) $C^*$-algebra is an associative algebra with a norm and an involution that acts as an adjoint operation. (An involution is an operation that is its own inverse.) For the purposes of this supplementary information, it suffices to simply consider the algebra formed by all of the $n$-by-$n$ matrices over the complex field (i.e., matrices whose entries can be complex numbers), $\mathbf{M}_{n}(\mathbb{C})$, with the usual operations like matrix multiplication, the conjugate transpose $M \mapsto M^\dagger$, in conjunction with the operator norm $\Vert \cdot \Vert$ that denotes the largest singular value of the matrix. This example of a $C^*$-algebra is particularly pertinent to physics, as it is the algebra that contains the physical observables on a finite Hilbert space of dimension $n$. (However, note that $\mathbf{M}_{n}(\mathbb{C})$ also contains matrices that are not Hermitian.) Regarding the naming of these algebras, the $*$ in $C^*$-algebra is $\dagger$ in physics, but we cannot call these $C^\dagger$-algebras so the nomenclature will clash with the notation. A real $C^*$-algebra possesses a second involution, $M\mapsto M^\rho$, that endows the algebra and its constituent elements with additional structure. For our purposes here of classifying fragile topology, this second involution is built from the symmetries that protect the system's fragile topology.

\subsection{Symmetries \label{sec:sym}}

Consider a square system with open boundary conditions. This means we have two position observables $X$ and $Y$ and Hamiltonian $H$.  If the system has sides of length $L$, then
$-L/2\leq X\leq L/2$ and  $-L/2\leq Y\leq L/2$ (i.e., all of the eigenvalues of $X$ and $Y$ fall within this range).
We will assume that $X$ and $Y$ commute as one assumes in material science that different positions observables are compatible. For simplicity, we now assume we have single-particle model on finite-dimensional Hilbert space $\mathfrak{H}$ and that this is a model of a closed system. Thus, we make the assumption that $H$ that is a Hermitian matrix $H = H^\dagger$. Finally, we expect some manner of locality, at least strong enough to imply $\| [H,X] \| \leq \delta$ and $\| [H,Y] \| \leq \delta$ for some $\delta > 0$ that seems small relative to the energy and length scales in the system (e.g., the lattice constants in the two directions and bulk spectral gap).

Fragile topology is generally associated with a system being symmetric with respect to the combination of a spatial symmetry and time-reversal. Here, we consider a time-reversal operator $\mathcal{T}$ that is antilinear and squares to $+I$ (spinless models) or to $-I$ (models incorporating spin). Moreover, we also want there to be also a unitary operator $C_2$ that implements rotation by 180$^\circ$ degrees. We will have $C_2^2 = I$ and assume $C_2$ is real (it is usually a permutation matrix). To keep the antilinear operators visually distinct from the familiar linear operators, we will use caligraphic font to denote such operators and proper composition notation $\mbox{-}\circ\mbox{-}$ where antiunitary operators are involved.

Altogether, we consider here models where $H$ has $C_2\circ\mathcal{T}$-symmetry but may lack either $C_2$-symmetry and/or $\mathcal{T}$-symmetry. Thus we introduce the composite operator
\begin{equation}
\mathcal{R} = C_2 \circ \mathcal{T}   ,  
\end{equation}
which will be an antiunitary symmetry, and our main assumption is
\begin{equation} 
    H \circ \mathcal{R} = \mathcal{R} \circ H. 
\end{equation}
Position and time should be independent, so we assume that $\mathcal{T}$ commutes with $X$ and $Y$.  As $C_2$ is rotation by a half-turn, we assume that $R$ anticommutes with both of these position operators. Thus we assume
\begin{equation}  \label{eqn:RT-symmetry-assumptions}
   \mathcal{R}  \circ X =-X \circ \mathcal{R} \enspace\mbox{and}\enspace  
    \mathcal{R}  \circ Y  = -Y \circ \mathcal{R}. 
\end{equation}

For the remainder, we will concentrate solely on the spinless case.  Further, we assume that time reversal is the standard choice of complex conjugation, i.e.,
\begin{equation}
\mathcal{T}(\mathbf{v}) = \overline{\mathbf{v}}    
\end{equation}
is a standing assumption from here on.  We will also use the notation $\mathcal{K}(\mathbf{v}) = \overline{\mathbf{v}}$ when convienient.

\subsection{Antiunitary symmetries and real $C^*$-algebras}

To understand the distance from a given system to an atomic limit,
we will need the mathematical techniques developed using real $C^*$-algebras over a series of papers from last decade 
\cite{loringSor2013AC_time_reversal,loringSor2014AC_real_orthogonal,loringSor2016AC_real_symmetric}. 
The relevant
real $C^*$-algebra here is the algebra of
all $2n$-by-$2n$ matrices $\mathbf{M}_{2n}=\mathbf{M}_{2n}(\mathbb{C})$, with an involution operation similar to the transpose that is built from $\mathcal{R}$.
We will use $\top$ denote transpose, so $M^\dagger = \overline{M^\top} = \overline{M}^\top$.  

To give $\mathbf{M}_{2n}$ a real structure, we define
a generalized involution $M\mapsto M^\rho$ by
\begin{equation*}
    M^\rho = \mathcal{R}^{-1} \circ M^\dagger \circ \mathcal{R}.
\end{equation*}
Since we have $\mathcal{R} \circ \mathcal{R} = I$ this simplifies to $M^\rho = \mathcal{R} \circ M^\dagger \circ \mathcal{R}$.
Recall that $\mathcal{R}=\mathcal{T}\circ C_2=C_2\circ\mathcal{T}$.  Since $\mathcal{T}$ denotes conjugation, we have
$\mathcal{R}(\mathbf{v})=C_2\overline{\mathbf{v}}$.
Therefore
\begin{align*}
\mathcal{R}\circ M^{\dagger}\circ\mathcal{R}(\mathbf{v}) & =\mathcal{R}\left(M^{\dagger}C_2\overline{\mathbf{v}}\right)\\
 & =C_2\overline{M^{\dagger}C_2\overline{\mathbf{v}}}\\
 & =C_2M^{\top}C_2\mathbf{v}
\end{align*}
and we find 
\begin{equation} \label{eqn:formula_for_rho}
M^{\rho}=C_2M^{\top}C_2.
\end{equation}
That is, this extra operation that creates a real structure for $\mathbf{M}_{2n}$ is the transpose intertwined with a
rotation.

\begin{rem}
Let us look at the very special case where the Hilbert space is $\mathbf{C}^2$, If there are two sites where rotation takes one location to the other, then $C_2$ is 
\begin{equation*}
C_2 = \left[\begin{array}{cc}
0 & 1\\
1 & 0
\end{array}\right].
\end{equation*}
The $\rho$ operation becomes
\begin{equation}
\left[\begin{array}{cc}
a & b\\
c & d
\end{array}\right]^{\rho}=\left[\begin{array}{cc}
d & b\\
c & a
\end{array}\right].
\end{equation}
Instead, if there is one site at the fixed point, say with two orbitals, then  $C_2$ will be
\begin{equation*}
C_2 = \left[\begin{array}{cc}
1 & 0\\
0 & 1
\end{array}\right],
\end{equation*}
and the $\rho$ operation becomes the standard transpose
\begin{equation}
\left[\begin{array}{cc}
a & b\\
c & d
\end{array}\right]^{\rho}=\left[\begin{array}{cc}
a & c\\
b & d
\end{array}\right].
\end{equation}
It is easy to forget how these behave differently and that the second case shows up as a sub-system whenever the full system has a site at the center of rotation.
\end{rem}

One axiom for real $C^*$-algebras is that the generalized involution, here $\rho$, needs to commute with the adjoint.  We know $C_2$ is Hermitian, so from Eq.~\eqref{eqn:formula_for_rho} we derive
\begin{equation*}
   \left( M^\rho \right)^\dagger  
   =  C_2\left(M^\top \right)^\dagger C_2 
    =  C_2\left(M^\dagger \right)^\top C_2  
    =  \left(M^\dagger \right)^\rho .   
\end{equation*}

Considering again a 2D system with Hamiltonian $H$ and position operators $X$ and $Y$, in real $C^*$-algebra terminology, we are assuming we have three Hermitian matrices, with $X$ and $Y$ commuting and $H$ almost commuting with the other two, with the symmetry assumptions
\begin{equation}
   X^\rho = -X, \quad    
    Y^\rho = -Y  \enspace\mbox{and}\enspace  
      H^\rho = H . 
\end{equation}
To clarify the symmetry on $H$, we have
\begin{equation*}
    H=\mathcal{R}\circ H\circ\mathcal{R}\implies H=\left(H^{\dagger}\right)^{\rho}\implies H=H^{\rho} 
\end{equation*}
since $H$ is assumed to be Hermitian.

\begin{defn}
Suppose $\mathcal{S}$ is an antiunitary symmetry on a Hilbert space $\mathfrak{H}$. An triple of  operators $(X,Y,H)$ on $\mathfrak{H}$ is said to have
$\mathcal{S}^{(-,-,+)}$-\emph{symmetry} if they are all Hermitian and
\begin{equation*}  
   \mathcal{S}  \circ X =-X \circ \mathcal{S}, \     
    \mathcal{S}  \circ Y  = -Y \circ \mathcal{S}  \ \mbox{and}\   
      \mathcal{S}  \circ H = H \circ \mathcal{S} . 
\end{equation*}
\end{defn}

\subsection{Converting from a standard physics setting to a standard math setting} 
\label{subsec:fixed_change_of_basis}

In a basis-free sense, any antiunitary that squares to $+I$ is equivalent to any other. Computer modeling of a physical system, however, needs to be done in
a fixed basis. To avoid decimating our physical intuition, we want to select those basis vectors so that they correspond to a single location.  To be able to utilize centuries of work in linear algebra, we want to work in a basis where we can work with the standard transpose
and not $M\mapsto M^\rho$.

For our computer algorithms, we will need an explicit unitary so that conjugation by that unitary converts $M\mapsto M^\top$ back to $M\mapsto M^\rho$.
This is essentially as described in Section 2 of \cite{boersemaL2015_K-theory_by_symmetries}, but this is an easy calculation so we include it here for completeness.

\begin{lem}
Suppose $R$ is real, unitary and $R^{2}=I$. Define
\begin{equation*}
M^{\tau}=RM^{\top}R.
\end{equation*}
Let 
\begin{equation*}
W=\frac{1}{\sqrt{2}}\left(R+iI\right).
\end{equation*}
This is unitary, and for any matrix $M$ we have
\begin{equation*} 
\left(WMW^{\dagger}\right)^{\top}=WM^{\tau}W^{\dagger}.
\end{equation*}
\end{lem}

\begin{proof}
The conditions on $R$ imply $R^{\top}=R$ and $R^{\dagger}=R$. To
see $W$ is unitary, we calculate
\begin{equation*}
2W^{\dagger}W=\left(R-iI\right)\left(R+iI\right)=2I.
\end{equation*}
Note that $W^{\top}=W$ and so $\overline{W}=W^{\dagger}$. Also
\begin{equation*}
iRW^{\dagger}=\frac{1}{\sqrt{2}}iR\left(R-iI\right)=W.
\end{equation*}
 The main calculation is then
\begin{align*}
\left(WMW^{\dagger}\right)^{\top} & =W^{\dagger}M^{\top}W
  =WRM^{\top}RW^{\dagger}
  =WM^{\tau}W^{\dagger}.
\end{align*}
\end{proof}

For working with the given physical system, we need a fixed $W$, namely
\begin{equation}\label{eqn:cannonical_W}
W=\frac{1}{\sqrt{2}}\left(C_2+iI\right).
\end{equation}
This is our choice to intertwine the two symmetry pictures, meaning
\begin{equation} \label{eqn:intertwine-R-with-conjugation}
\left(WMW^{\dagger}\right)^{\top}=WM^{\rho}W^{\dagger}.
\end{equation}
Indeed, there are other choices for $W$ to achieve this intertwining and we need a fixed choice to get a well-defined local topological invariant.

The basic observables in the physical models we wish to study will be triples in the following set.
\begin{multline} 
\mathcal{M}_{\epsilon}(\mathcal{R},2n)= 
\Biggl\{ \left(X,Y,H\right)\in\left(\mathbf{M}_{2n}\right)^{3} \\ 
\left| \begin{aligned} & (X,Y,H)\text{ is }\ensuremath{\mathcal{R}^{(-,-,+)}}\text{-symmetric}\\
 & \left\Vert [H,X]\right\Vert \leq\epsilon,\ \left\Vert [H,Y]\right\Vert \leq\epsilon\\
 & L(X,Y,H)\mbox{ is invertible}
\end{aligned}
\right\}  . \label{eq:s11}
\end{multline}
If the Hamiltonian actually commutes with position then $(X,Y,H)\in \mathcal{M}_{0}(\mathcal{R},2n)$ is \emph{an atomic limit} in this class of models.

After a change of basis, so conjugation of everything by $W$, this set becomes the following,
\begin{multline} 
\mathcal{M}_{\epsilon}(\mathcal{K},2n)=\Biggl\{ \left(X,Y,H\right)\in\left(\mathbf{M}_{2n}\right)^{3} \\
\left| \begin{aligned} & (X,Y,H)\text{ is }\ensuremath{\mathcal{K}^{(-,-,+)}}\text{-symmetric}\\
 & \left\Vert [H,X]\right\Vert \leq\epsilon,\ \left\Vert [H,Y]\right\Vert \leq\epsilon\\
 & L(X,Y,H)\mbox{ is invertible}
\end{aligned}
\right\}  . \label{eq:s12}
\end{multline}
Except for how the relevant matrices act on vectors, these two sets have identical structure.  They are isometric as metric spaces, for example.

In Eqs.~\eqref{eq:s11} and \eqref{eq:s12}, $L(X,Y,H)$ refers to the spectral localizer formed from the triplet of matrices. This operator will be formally introduced in Sec.~\ref{sec:p1}.

\subsection{The commutative case to model atomic limits -- part I \label{sec:p1}}

A system is \emph{in an atomic limit} if the Hamiltonian commutes with all of the position observables.  This shuts down hopping terms between any neighboring sites, so an excitation at a site evolves as if the other sites do not exist. (Different orbitals at the same lattice site can still be coupled in an atomic limit.) In the next subsection we classify all such commuting, locally gapped systems, i.e.\ elements of $\mathcal{M}_{0}(\mathcal{R},2n)$.  Here we first examine the structure of $\mathcal{M}_{0}(\mathcal{K},2n)$.

We need to understand the structure of three commuting Hermitian matrices, two of which are purely imaginary and the third is real. (Again, as $\mathcal{K}$ is just complex conjugation, and $(X,Y,H)$ is $\mathcal{K}^{(-,-,+)}$-symmetric in $\mathcal{M}_{0}(\mathcal{K},2n)$, the position operators are purely imaginary. This is not physically relevant in this form, but it is a convenient basis to explore some of the mathematics.) As such, we need work with a real version of the spectral theorem, and identify an invariant that makes sense for a mix of real and imaginary matrices.  

\begin{thm} \label{thm:real_spectral_thm}
If $N_{1},\dots,N_{k}$ are commuting normal real $n$-by-$n$ matrices,
then there is a real orthogonal matrix $U$, of determinant one, so that
\begin{equation*}
N_{j}=UD_{j}U^{\dagger}
\end{equation*}
with each $D_{j}$ block diagonal, with every block either $1$-by-$1$
and real or a $2$-by-$2$ block of the form
\begin{equation*}
\left[\begin{array}{cc}
a & -b\\
b & a
\end{array}\right]
\end{equation*}
with $a$ and $b$ real.
\end{thm}

\begin{proof}
Except for the requirement on the determinant of $U$, this can be found on page p.\ 292, Theorem 12, of \cite{Gantmacher1959matrix_theory}.
It also follows easily from the version
of the Schur decomposition that applies to commuting real normal matrices.
If $\det(U)=-1$ then we can multiply its first row by $-1$, and multiply the first row and column of each $X_j$ by $-1$, to get this factorization with the determinant of orthogonal matrix having the opposite sign.
\end{proof}

\begin{thm} \label{thm:imag-imag-real-spectral-theorem}
If $X$, $Y$ and $H$ are commuting $2n$-by-$2n$ Hermitian matrices, with $X$  and $Y$ purely imaginary and $H$ real,  then there is a real orthogonal matrix $U$ of determinant one so that
\begin{equation*}
X=UX_DU^{\dagger},\  
Y=UY_DU^{\dagger},\ 
H=UH_DU^{\dagger},\ 
\end{equation*}
where  $H_D$ is a real diagonal matrix  and $X_D$ and $Y_D$ are block diagonal with
$2$-by-$2$ blocks of the form
\begin{equation}\label{eqn:imaginary_Hermitian_block}
\left[\begin{array}{cc}
0 & -ib\\
ib & 0
\end{array}\right]
\end{equation}
 with $b$ real.
\end{thm}


\begin{proof}
Notice that $iX$, $iY$ and $H$ are normal and real.  Applying Theorem~\ref{thm:real_spectral_thm} to these three provides a unitary $U$ of norm one so that
\begin{align*}
iX  =UX_{1}U^{\dagger}, \quad
iY  =UY_{1}U^{\dagger} \enspace\mbox{and}\enspace 
H  =UH_{D}U^{\dagger}
\end{align*}
with being $H_{D}$ real and diagonal and $X_{1}$ and $Y_{1}$ being
real block diagonal with blocks of the form
\begin{equation*}
\left[\begin{array}{cc}
a & -b\\
b & a
\end{array}\right].
\end{equation*}
Since $H$ is Hermitian, $H_{D}$ will be Hermitian. Let $X_{D}=-iX_{1}$
and $Y_{D}=-iY_{1}$. Since $iX$ and $iY$ are be anti-Hermitian, the
blocs of $X_{1}$ and $Y_{1}$ will be anti-Hermitian. The blocks
of $X_{D}$ and $Y_{D}$ will be Hermitian and purely imaginary, so
will be of the form in Eq.~\eqref{eqn:imaginary_Hermitian_block}.
\end{proof}

In Eq.~\eqref{eq:s12}, we noted that the spectral localizer for $(X,Y,H)$ must be invertible.
We now pin down the choice of the $\gamma_j$ in the definition of the spectral localizer in $\mathcal{M}_{\epsilon}(\mathcal{K},2n)$.  We use the Pauli spin matrices in a specific order and so we will use the convention
\begin{equation*}
    L(X,Y,H) = X\otimes \sigma_x +  Y\otimes \sigma_z +  H\otimes \sigma_y .
\end{equation*}
With this convention, $L(X,Y,H)$ will be skew-symmetric (as well as Hermitian) and so has a well-defined Pfaffian.  The sign of this Pfaffian is a natural invariant. 
The fact that the spectral  localizer is skew-symmetric forces its eigenvalues to appear in pairs $\pm \alpha$.  The spectral localizer is $4n$-by-$4n$ so there will be an even number of pairs and so its determinant will be positive.  This implies that its Pfaffian will be real so the sign makes sense.

More generally, we define
\begin{multline}
    L_{(x,y,E)}(X,Y,H) = \\
    (X-x)\otimes \sigma_x +  (Y-y)\otimes \sigma_z +  (H-E)\otimes \sigma_y .
\end{multline}
Only if $x=y=0$ do we retain the symmetries needed to give us a skew-symmetric localizer, so our local invariant is only defined at the center of rotation.

Small examples show that the Pfaffian can come out positive or negative.  The sign of the Pfaffian cannot change along a path in  $\mathcal{M}_{\epsilon}(\mathcal{K},2n)$ since we have excluded the case where the spectral localizer is singular. 

\begin{rem}
A $2n$-by-$2n$ matrix $M$ that is that is are both Hermitian and
skew-symmetric must have spectrum whose eigenvalues come in pairs
$\pm\lambda$ \cite{Lax2007Matrices,Serre2010Matrices}. Two invertible Hermitian, skew-symmetric matrices $H_{0}$
and $H_{1}$ can be connected by a path of invertible Hermitian, skew-symmetric
matrices if an only if $\textup{Pf}(H_{0})$ and $\textup{Pf}(H_{1})$
are of the same sign. Given a path $H_{t}$ of Hermitian, skew-symmetric
matrices, the sign of $\textup{Pf}(H_{t})$ can only change when $\textup{Pf}(H_{t})=\pm\sqrt{\left(\det(H_{t})\right)}$
is zero. This means the sign of $\textup{Pf}(H_{t})$ can only change
when a pair of eigenvalues crosses $0$ in opposite directions.
\end{rem}

\begin{example} \label{ex:opo_sites_math_picture}
For any real $\alpha$, $\beta$ and $\gamma$, consider the matrices
\begin{equation*}
X=\left[\begin{array}{cc}
0 & -i\alpha\\
i\alpha & 0
\end{array}\right],\ Y=\left[\begin{array}{cc}
0 & -i\beta\\
i\beta & 0
\end{array}\right],\ H=\left[\begin{array}{cc}
\gamma & 0\\
0 & \gamma
\end{array}\right].
\end{equation*}
These commute, are Hermitian, with the first two antisymmetric and
the last symmetric. We find
\begin{equation*}
L(X,Y,H)=\left[\begin{array}{cccc}
0 & -i\beta & -i\gamma & -i\alpha\\
i\beta & 0 & i\alpha & - i\gamma\\
i\gamma & -i\alpha & 0 & i\beta\\
i\alpha &  i\gamma & -i\beta & 0
\end{array}\right]
\end{equation*}
whose Pfaffian is always positive, as it equals
\begin{equation*}
\left(-i\beta\right)\left(i\beta\right)-\left(-i\gamma\right)\left(-i\gamma\right)+\left(-i\alpha\right)\left(i\alpha\right)=\alpha^{2}+\beta^{2}+\gamma^{2}.
\end{equation*}
Since
\begin{align*}
    \left(L(X,Y,H) \right)^2 &= \left(X^2 + Y^2 + H^2 \right) \otimes I_2 \\
    &= (\alpha^2 + \beta^2 + \gamma^2 ) \otimes I_4 
\end{align*}
this must have spectrum contained in the set $\{ \pm (\alpha^2 + \beta^2 + \gamma^2 ) \}$.  Every skew-symmetric, Hermitian matrix has spectrum that is symmetric across $0$, so  both eigenvalues must have multiplicity two.  In particular, $L(X,Y,H)$ is invertible so long as at least one of $\alpha$, $\beta$ or $\gamma$ is nonzero.
\end{example}

\begin{example}
For any real numbers  $\gamma_j$, consider the matrices
\begin{equation}
X=\left[\begin{array}{cc}
0 & 0\\
0 & 0
\end{array}\right],\ Y=\left[\begin{array}{cc}
0 & 0\\
0 & 0
\end{array}\right],\ H=\left[\begin{array}{cc}
\gamma_1 & 0\\
0 & \gamma_2
\end{array}\right].
\end{equation}
As the off-diagonal elements are set to zero, these commute.  Trivially, $X$ and $Y$ are purely imaginary, and $H$ is real.  All are Hermitian.  Now we find that $L(X,Y,H)$ can have negative  Pfaffian, as it equals
\begin{equation*}
-\left(-i\gamma_1\right)\left(-i\gamma_2\right)=\gamma_1\gamma_2.
\end{equation*}
We find that  $L(X,Y,H)$ is invertible so long as the $\gamma_j$ are both nonzero.  Indeed, its spectrum is $\{ \pm \gamma_1, \pm \gamma_2  \}$.
\end{example}

Now we show that the sign of the Pfaffian is the only 
obstruction to connecting two triples in  $\mathcal{M}_{0}(\mathcal{K},2n)$.

\begin{prop}
If $(X_1,Y_1,H_1)$  and  $(X_2,Y_2,H_2)$ are in  $\mathcal{M}_{0}(\mathcal{K},2n)$, 
then these are connected by a path in   $\mathcal{M}_{0}(\mathcal{K},2n)$ if, and only if, 
\begin{equation*}
\textup{sign}\left(\textup{Pf}\left(L(X_{1},Y_{1},H_{1}\right)\right) =
\textup{sign}\left(\textup{Pf}\left(L(X_{2},Y_{2},H_{2}\right)\right)
\end{equation*}
\end{prop}

\begin{proof}
  Since $\textup{SU}(n)$ is connected, we can assume both triples are block diagonal.   In each $2$-by-$2$ common block, the triple looks like
\begin{equation*}
\left[\begin{array}{cc}
0 & -i\alpha\\
i\alpha & 0
\end{array}\right],\ \left[\begin{array}{cc}
0 & -i\beta\\
i\beta & 0
\end{array}\right],\ \left[\begin{array}{cc}
\gamma_1 & 0\\
0 & \gamma_2
\end{array}\right].
\end{equation*}
The commutativity assumption implies that either $\alpha=\beta=0$ or $\gamma_1  =\gamma_2$.   The blocks with $\gamma_1=\gamma_2$ can all be connected by a path of such blocks to the special case where $\alpha=\beta=0$.  Thus we can assume $X=Y=0$ and $H$ is diagonal.

Any common block of the form
\begin{equation*}
\left[\begin{array}{cc}
0 & 0\\
0 & 0
\end{array}\right],\ \left[\begin{array}{cc}
0 & 0\\
0 & 0
\end{array}\right],\ \left[\begin{array}{cc}
-1 & 0\\
0 & -1
\end{array}\right]
\end{equation*}
is homotopic to
\begin{equation*}
\left[\begin{array}{cc}
0 & 0\\
0 & 0
\end{array}\right],\ \left[\begin{array}{cc}
0 & 0\\
0 & 0
\end{array}\right],\ \left[\begin{array}{cc}
1 & 0\\
0 & 1
\end{array}\right],
\end{equation*}
as we see from the path 
\begin{equation*}
\left[\begin{array}{cc}
0 & 0\\
0 & 0
\end{array}\right],\ \left[\begin{array}{cc}
0 & -i\sin\theta\\
i\sin\theta & 0
\end{array}\right],\ \left[\begin{array}{cc}
\cos\theta & 0\\
0 & \cos\theta
\end{array}\right].
\end{equation*}
Since 
\begin{equation*}
\left[\begin{array}{cc}
0 & -1\\
1 & 0
\end{array}\right]\left[\begin{array}{cc}
1 & 0\\
0 & -1
\end{array}\right]\left[\begin{array}{cc}
0 & 1\\
-1 & 0
\end{array}\right]^{\dagger}=\left[\begin{array}{cc}
-1 & 0\\
0 & 1
\end{array}\right]
\end{equation*}
and 
\begin{equation*}
\left[\begin{array}{cc}
0 & -1\\
1 & 0
\end{array}\right]
\end{equation*}
is homotopic to $I_{2}$ in $\textup{SO(2)}$, we can assume blocks
all look like $0$, $0$ followed by
\begin{equation*}
\left[\begin{array}{cc}
1 & 0\\
0 & -1
\end{array}\right]\text{ or }\left[\begin{array}{cc}
1 & 0\\
0 & 1
\end{array}\right].
\end{equation*}
Similarly, 
\begin{multline}
\left[\begin{array}{cccc}
0 & 0 & -1 & 0\\
0 & 1 & 0 & 0\\
1 & 0 & 0 & 0\\
0 & 0 & 0 & 1
\end{array}\right]\left[\begin{array}{cccc}
-1 & 0 & 0 & 0\\
0 & 1 & 0 & 0\\
0 & 0 & -1 & 0\\
0 & 0 & 0 & 1
\end{array}\right]\left[\begin{array}{cccc}
0 & 0 & -1 & 0\\
0 & 1 & 0 & 0\\
1 & 0 & 0 & 0\\
0 & 0 & 0 & 1
\end{array}\right]^{\dagger}= \\
\left[\begin{array}{cccc}
1 & 0 & 0 & 0\\
0 & 1 & 0 & 0\\
0 & 0 & 1 & 0\\
0 & 0 & 0 & 1
\end{array}\right] \notag
\end{multline}
so we can assume there is at most one block with opposite signs. We
can use a unitary of determinent one to swap that block around, so every
triple is homotopic to either $(0,0,I)$ or $(0,0,D)$ where $D$
is diagonal with all diagonal elements equal to $1$ except the top-left
element which is equal to $-1$. 

\end{proof}


\subsection{The commutative case to model atomic limits -- part II}

With these results in hand for the non-physical $\mathcal{M}_{\epsilon}(\mathcal{K},2n)$, we are in a position to translate these results to the physically meaningful $\mathcal{M}_{\epsilon}(\mathcal{R},2n)$.

\begin{defn}
Suppose $(X,Y,H)$ is an element $\mathcal{M}_{\epsilon}(\mathcal{R},2n)$.  We define the spectral localizer as
\begin{multline*}
L_{(x,y,E)}\left(X,Y,H\right)= \\
\left(X-xI\right)\otimes\sigma_{x}+\left(Y-yI\right)\otimes\sigma_{z}+\left(H-EI\right)\otimes\sigma_{y}
\end{multline*}
 and, when $x=y=0$, we can define a local invariant, taking values in $\{1,-1\}\cong \mathbb{Z}_2$,
\begin{multline}\label{eqn:local_index_formula}
\zeta_{E}(X,Y,H) = \\
\textup{sign}\left(\textup{Pf}\left(L_{(x,y,E)}\left(WXW^{\dagger},WYW^{\dagger},WHW^{\dagger}\right)\right)\right),
\end{multline}
where $W$ is as defined in Eq.~\eqref{eqn:cannonical_W}.
This is defined whenever the local Clifford gap is non-zero and $x=y=0$, where the local Clifford gap
is defined as
\begin{equation*}
\mu^{\textup{C}}_{(x,y,E)}(X,Y,H)=\sigma_{\min}\left(L_{(x,y,E)}(X,Y,H)\right),
\end{equation*}
i.e., the smallest singular value of the spectral localizer.
\end{defn}

In passing, we note that a more aesthetically pleasing formula would be
\begin{multline}\label{eqn:local_index_a;t}
\zeta_{E}\left(X,Y,H\right)=\\
\textup{sign}\left(\textup{Pf}\left((W\otimes I)L_{(0,0,E)}\left(X,Y,H\right)(W\otimes I)^\dagger\right)\right).
\end{multline}
However, the formula in Eq.~\eqref{eqn:local_index_formula} will generally yield a slightly faster numerical algorithm. 

We now consider a pair of examples to again show how $\zeta_{E}(X,Y,H)$ can take values of both $\pm 1$.

\begin{example}\label{ex:opo_sites_physics_picture}
Assume we have just two distinct sites that are swapped by rotation, at locations $(\pm\alpha,\pm\beta)$ with $\alpha\neq 0 $ or $\beta\neq 0 $. If
\begin{equation*}
X =\left[\begin{array}{cc} -\alpha & 0\\ 0 & \alpha \end{array}\right],\ 
Y =\left[\begin{array}{cc} -\beta & 0\\ 0 & \beta \end{array}\right],\ 
H =\left[\begin{array}{cc} \gamma & 0\\ 0 & \gamma \end{array}\right] 
\end{equation*}
then $WXW^\dagger$,  $WYW^\dagger$ and  $WHW^\dagger$are the matrices discussed in Example~\ref{ex:opo_sites_math_picture}.
These form a triple in $ \mathcal{M}_{0}(\mathcal{R},2)$ and so
\begin{equation*}
\zeta_{E}(X,Y,H)=1
\end{equation*}
for any $E\neq \gamma$.
\end{example}

\begin{example}
Let    
\begin{equation*}
X=\left[\begin{array}{cc}
0 & 0\\
0 & 0
\end{array}\right],\ Y=\left[\begin{array}{cc}
0 & 0\\
0 & 0
\end{array}\right],\ H=\left[\begin{array}{cc}
\alpha & -i\beta\\
i\beta & \alpha
\end{array}\right].
\end{equation*}
These constitute a triple in $ \mathcal{M}_{0}(\mathcal{R},2)$.  We find
\begin{multline*}
WXW^{\dagger}=\left[\begin{array}{cc}
0 & 0\\
0 & 0
\end{array}\right],\ WYW^{\dagger}=\left[\begin{array}{cc}
0 & 0\\
0 & 0
\end{array}\right], \\ 
WHW^{\dagger}=\left[\begin{array}{cc}
\alpha+\beta & 0\\
0 & \alpha-\beta
\end{array}\right]
\end{multline*}
and so 
\begin{equation*}
\zeta_{E}(X,Y,H)=\textrm{sign}\left((\alpha+\beta-E)(\alpha-\beta-E) \right).    
\end{equation*}
This equals $-1$ for $E$ between $\alpha-\beta$ and $\alpha+\beta$.
\end{example}

\subsection{Connecting to an Atomic limit \label{subsec:Flattening}}

We discussed earlier that two locally-gapped, $C_{2}\mathcal{T}$-symmetric systems of different invariant cannot be connected by a continuous path of
such systems. Here we discuss a possible  converse. We already know that atomic limits of the same invariant can be connected, so will attempt to show that every locally-gapped $C_{2}\mathcal{T}$-symmetric system
can be connected to an atomic limit where this limit system is also locally-gapped and $C_{2}\mathcal{T}$-symmetric. A complication arises; if the starting system $(X_{0},Y_{0},H_{0})$ is in 
$\mathcal{M}_{\epsilon}(\mathcal{R},2n)$,
and $(X_{1},Y_{1},H_{1})$ is that atomic limit, we cannot expect
$(X_{t},Y_{t},H_{t})$ to stay in $\mathcal{M}_{\epsilon}(\mathcal{R},2n)$.
Instead, we hope to prove that $(X_{t},Y_{t},H_{t})$ stays in in $\mathcal{M}_{\delta}(\mathcal{R},2n)$
where $\delta$ is a bit larger than $\epsilon$. 

Here we offer only a sketch of a possible argument. The estimates on how large $\delta$ will be will depend on both $\epsilon$ and the size of the local gap.  We anticipate that soon someone will develop rigorous results about paths to commuting matrices with these symmetries, for both open and periodic boundary conditions.  The recent advances \cite{herrera2024SymmetryBootstrap} in matrix approximations that respect antilinear and linear symmetries should work in this setting.  Earlier work, not involving antilinear symmetries, indicates that this approximation only works when the ``joint spectrum'' of the given matrices is two-dimesional \cite{EndersShulman2023AC_matrices_and_dimension}.  It is for this reason we start with spectral flattening.

There are at least two ways to define a local gap of system $(X,Y,H)$.
The quadratic gap is the square root of the smallest singular value of 
\begin{equation*}
Q(X,Y,Z)=X^{2}+Y^{2}+H^{2},
\end{equation*}
while the Clifford gap is the smallest singular value of
\begin{equation*}
L(X,Y,H)=\left[\begin{array}{cc}
H & X-iY\\
X+iY & -H
\end{array}\right].
\end{equation*}
The Clifford gap at zero is denoted by $\mu_{(0,0,0)}(X,Y,H)$ in the
main text.  Since
$L(X,Y,H)$ is an approximate square-root of an amplification of $Q(X,Y,H)$,
these two notions of local gap are approximately equal. 

The goal of the spectral flattening is to gain the approximate relation
which makes this an approximate representation of a sphere. We expect this reduction to help, as many related commuting-matrix approximation problems require at most a two-dimensional ``joint spectrum'' \cite{EndersShulman2023AC_matrices_and_dimension}.

Suppose we are given $(X,Y,H)$ in $\mathcal{M}_{\epsilon}(\mathcal{R},2n)$.
Beyond insisting that $L(X,Y,H)$ is invertible, let's assume
\begin{equation*}
g^2\leq\left(L(X,Y,H)\right)^{2}\leq G^2
\end{equation*}
where $2\epsilon<g<G$. Then
\begin{equation*}
g^2-2\epsilon\leq X^{2}+Y^{2}+H^{2}\leq G^2+2\epsilon .
\end{equation*}
The usual spectral flatting adjusts only $H$, replacing it with
$\tilde{H}=H\left(H^{2}\right)^{-\frac{1}{2}}$.
We cannot do this here, as open boundary conditions means there is
no sizable gap expected in our $H$ and so the commutator of $\tilde{H}$
with $X$ and $Y$ will likely blow up. Instead, we define 
\begin{equation*}
Q=X^{2}+Y^{2}+H^{2}
\end{equation*}
and then
\begin{equation*}
X_{t}=Q^{-\frac{t}{4}}XQ^{-\frac{t}{4}},\ Y_{t}=Q^{-\frac{t}{4}}YQ^{-\frac{t}{4}},\ H_{t}=Q^{-\frac{t}{4}}HQ^{-\frac{t}{4}}.
\end{equation*}
These are continuous paths of Hermitian matrices. We next check that
we still have $C_{2}\mathcal{T}$-symmetry. We can show, by polynomial approximation
to the square-root function, that
$\left(Q^{-\frac{t}{4}}\right)^{\rho}=\left(Q^{\rho}\right)^{-\frac{t}{4}}$
and so 
\begin{equation*}
\left(Q^{-\frac{t}{4}}\right)^{\rho}=\left(X^{2}+Y^{2}+H^{2}\right)^{-\frac{t}{4}}=Q^{-\frac{t}{4}}.
\end{equation*}
Given this, it is easy to show that
\begin{equation*}
    X_t ^\rho = -X_t ^\rho,\ Y_t ^\rho = -Y_t ^\rho,\ H_t ^\rho = H_t ^\rho.
\end{equation*}
Again using polynomial approximation, one can show that 
\begin{equation*}
\left\Vert \left[H_{t},X_{t}\right]\right\Vert \leq\delta,\ \left\Vert \left[H_{t},Y_{t}\right]\right\Vert \leq\delta
\end{equation*}
where $\delta$ is large than $\epsilon$, but only depends on $\epsilon$
and $g$. See Theorem 3.2.32 of \cite{bratteli2012OperAlg_StatMec}.
Thus $(X_{t},Y_{t},H_{t})$ is a continuous path in in $\mathcal{M}_{\delta}(\mathcal{R},2n)$.
Since
\begin{equation*}
X_{t}^{2}\approx Q^{-\frac{t}{2}}X^{2}Q^{-\frac{t}{2}}
\end{equation*}
etc., we find 
\begin{equation*}
\left(L(X_{t},Y_{t},H_{t})\right)^{2}\approx Q^{1-t}.
\end{equation*}
Here we will need $\epsilon$ small compared to $g^2$ to ensure that
$(X_{t},Y_{t},H_{t})$ remains gapped for every $t$. 

At $t=1$ we
find we have a new approximate relation,
\begin{equation*}
X_{1}^{2}+Y_{1}^{2}+H_{1}^{2}\approx I.
\end{equation*}
As we explain in Section A of \cite{Cerjan2024a}  we can map the coordinate
functions in $C(S^{2})$ to these three matrices and extend to get
a map 
\begin{equation*}
\varphi:C(S^{2})\rightarrow\boldsymbol{M}_{2n}(\mathbb{C})
\end{equation*}
that behaves somewhat like a $*$-homomorphism. In this case, the
real structure $M\mapsto M^{\rho}$ will correspond to the real structure
on $C(S^{2})$ induced by a 180-degree rotation. This is very close
to the setting of \cite{loringSor2014AC_real_orthogonal}. The main result there tells us that two
almost commuting real orthogonal matrices are always close to commuting
real orthogonal matrices. What we have is a similar mathematical situation,
where we replace the role of the two-torus with a rotation by the two-sphere with a rotation. If the
following conjecture is correct, we can derive the same result, but
with $X^{\rho}=-X$, etc., as we have the unitary $W$ that
conjugates one picture to the other.

\begin{conj}
For any $\eta>0$ there is there is a $\delta>0$ such that, for all
$n$, given matrices $X_1$, $Y_1$ and $H_1$ in $\boldsymbol{M}_{2n}(\mathbb{C})$
with 
\begin{equation*}
X_1^{\dagger}=X_1=-\overline{X_1},\ Y_1^{\dagger}=Y_1=-\overline{Y_1},\ H_1^{\dagger}=H_1=\overline{H_1}
\end{equation*}
and 
\begin{equation*}
\left\Vert \left[H_1,X_1\right]\right\Vert \leq\delta,\ \left\Vert \left[H_{1},Y_{1}\right]\right\Vert \leq\delta,\left\Vert X_{1}^{2}+Y_{1}^{2}+H_{1}^{2}-I\right\Vert \leq\delta,
\end{equation*}
there there is a triple $(X_{2},Y_{2},H_{2})$ of commuting Hermitian
matrices, with $X_{2}^{\rho}=-X_{2}$, $Y_{2}^{\rho}=-Y_{2}$, $H_{2}^{\rho}=H_{2}$
and $X_{2}^{2}+Y_{2}^{2}+H_{2}^{2}=I$ and 
\begin{equation*}
\left\Vert X_{2}-X_{1}\right\Vert \leq\eta,\ \left\Vert Y_{2}-Y_{1}\right\Vert \leq\eta,\ \left\Vert H_{2}-H_{1}\right\Vert \leq\eta.
\end{equation*}
\end{conj}

Assuming this conjecture true, we find that we can find a commuting
$C_{2}\mathcal{T}$-symmetric system $(X_{2},Y_{2},H_{2})$ close to $(X_{1},Y_{1},H_{1})$.
As long as $\eta$ is small, we can linearly interpolate between $(X_{1},Y_{1},H_{1})$
and $(X_{2},Y_{2},H_{2})$ to complete the path. The result is a path
of locally gapped, $C_{2}\mathcal{T}$-symmetric systems with limited growth
in the commutators that terminates in an atomic limit. By limited
growth we mean 
\begin{equation*}
\left\Vert \left[H_{t},X\right]\right\Vert \leq\frac{2\left\Vert \left[H_{t},X\right]\right\Vert }{g}
\end{equation*}
 or something similar.

\subsection{Strong vs Fragile topology in the pseudospectrum}

The invariant $\zeta_{E}$ differs fundamentally from the local version
of the Chern number that looks at 
\begin{equation*}
C_{L}(x,y,E)=\frac{1}{2}\textup{sig}\left(L_{(x,y,E)}(X,Y,H)\right).
\end{equation*}
The key difference is that our new invariant can only defined at $x=y=0$. Because $C_{L}(x,y,E)$ is
not restricted to $x=y=0$, the Clifford spectrum (points in position-energy
space where the local gap is zero) of a Chern insulator always is
sphere-like. Specifically, if one travels on any ray out of the origin
(assuming a local gap at the origin) one will hit at least one point
in the Clifford spectrum. In math terms, we say that the Clifford
spectrum separates the origin from infinity. 

In Fig.~\ref{fig:figsm2} we plot a portion near the origin of the $y=0$ slice
of the Clifford pseudospectrum of the tight-binding model considered in
the main text next to the same for a Chern insulator (the Haldane model).  Of particular importance is the portion that is completely black.  This set is called the Clifford spectrum. At first glance, the black parts in the both plots look sphere-like, but a close look reveals a difference. The black region in Fig.~\ref{fig:figsm2}(b) does not entirely enclose
the origin; one can see separation between some of the black
dots. Indeed, we suspect that the unusual spectral flattening discussed
in \S \ref{subsec:Flattening} may not be necessary. That is, the speckled look of the Clifford pseudospectrum seems to indicate that the TBG tight binding model is already close to an atomic limit.

Local invariants that detect other strong invariants, at least in classes A, AI and AII, also force the Clifford spectrum to separate the origin from infinity.  For example, see Figure 9.8 in \cite{Loring2015} for a horizontal slice ($E=0$) of a spin-Chern insulator. For a 3D example in class AII see Fig. 10.1 \cite{Loring2015}, where now the image is only for $z=E=0$.  For a 4D example in class AI, see Fig. 7.1 \cite{cerjanLoring2024even_spheres}.  In all cases, the Clifford spectrum has a portion that surrounds the origin like a sphere, and this feature is stable under perturbations within the symmetry class.

\begin{figure}[h]
\center
\includegraphics[width=\columnwidth]{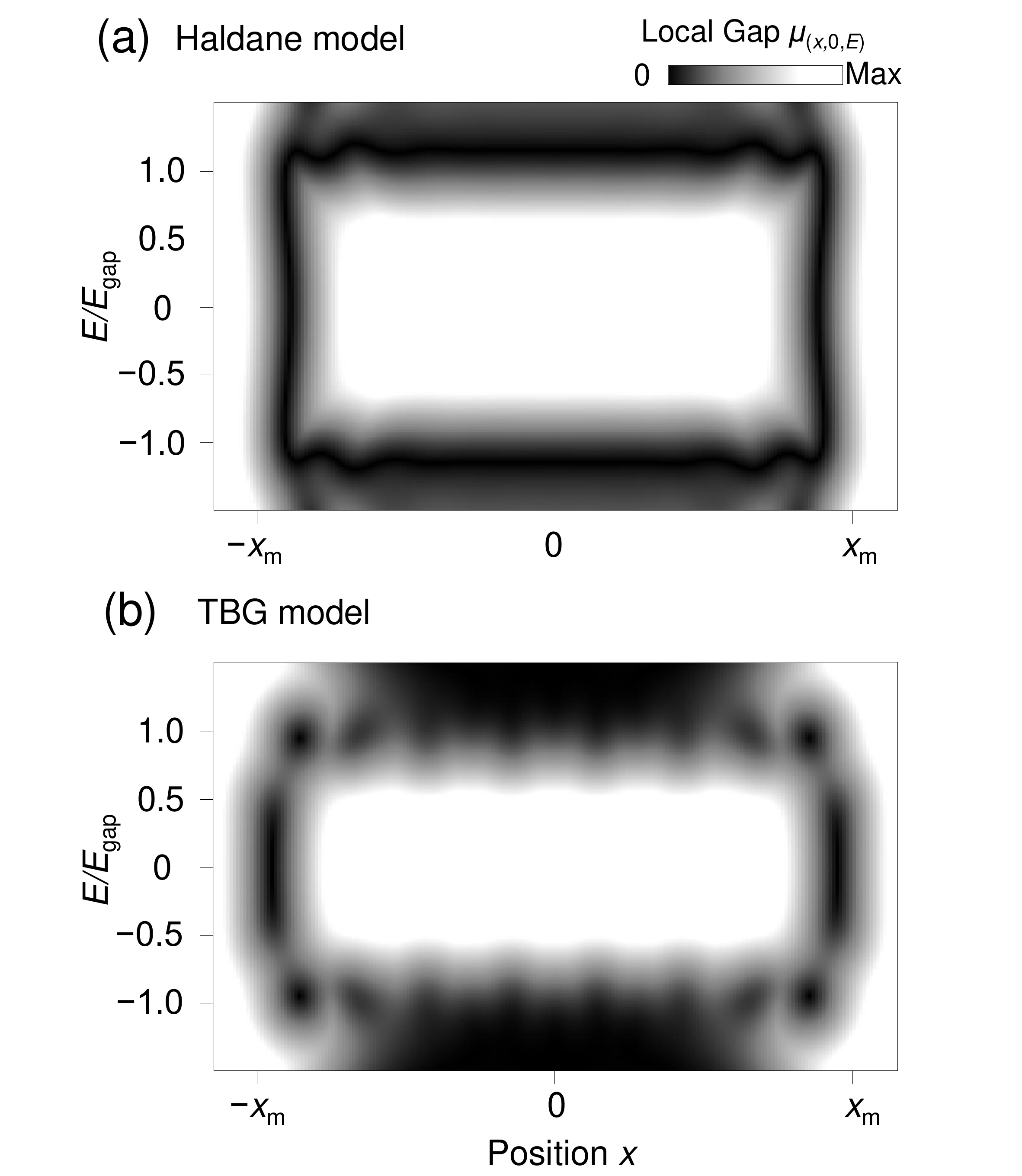}
\caption{
Clifford pseudospectrum -- strong vs.\ fragile topology. 
Comparison of the Clifford pseudospectrum between models of a Chern insulator and a fragile insulator. 
\text{(a)} A standard Haldane model of a Chern insulator, with local gap shown along the slice $y=0$.
\text{(b)} The TBG tight-binding model considered in main text, same slice shown of the local gap. The blackest regions in each is Clifford spectrum, where the local gap is zero.  In panel \text{(a)} we see a slice of a sphere-like region that completely surrounds the position-energy origin $(x,y,E)=(0,0,0)$.  In panel \text{(b)} the black region does not completely enclose the origin.
}
\label{fig:figsm2}
\end{figure}

\subsection{An explanation via $K$-theory}

The formula in Equation~\eqref{eqn:local_index_formula} for the invariant $\zeta_0$ could have been deduced from Theorem 4.5 of \cite{Boersema2020_K-theory_by_symmetriesII}.  That theorem is about explicit generators of the various $KO$ and $KU$ groups of 
$C(S^{2}\setminus\{\textbf{np}\},\sigma)   $
where $\sigma$ is the rotation by 180 degrees that fixes the removed North pole $\textbf{np}$ and the South pole. This tells us that
$L\left(\hat{x},\hat{y},\hat{z}\right)
$
is a unitary with the correct symmetries to represent a generator of 
\begin{equation*}
KO_2\left(C(S^{2}\setminus\{\textup{np}\},\sigma)\right),
\end{equation*}
so long as the $\gamma$ matrices are chosen correctly. Here $\hat{x}$ and so forth are the coordinate functions restricted to the sphere.
Notice, however, that all the formulas in this table look like a localizer of the coordinate functions with varying choices of $\gamma$ matrices.
As such, it is often possible to guess the needed index formulas, just by seeking the minimum size $\gamma$ matrices that have the needed symmetries so the spectral localizer will end up with the expected
symmetries.

To get an invariant out an element of a $KO$-group of $C(S^d)$ with some symmetry, we just replace the coordinate functions of with the matrix observables and arrive at in invertible element in $\mathbf{M}_n(\mathbb{C})$ with some antiunitary symmetry.  (In some cases more than one antiunitary symmetry is involved.)  Since the target algebra does not change as we change dimensions and symmetry class, all these invariants (so far) end up with one of three calculations, the signature, sign of a determinant, or sign of a Pfaffian.  The mathematical formalism here involves universal $C^*$-algebras, as explained in Section 3 of \cite{loring2014quantitative}.  Fortunately, one can understand why these indices can only change with the local gap closes without this abstraction.

There are explicit generators of $KO$ groups calculated in \cite{boersemaL2015_K-theory_by_symmetries} for a different symmetries on the two-sphere, and research in this is ongoing.  It is anticipated that pseudospectral methods to create local topological invariants will work with antiunitary symmetries that commute with some of the position observables and anticommuting with others, and either commute or anticommute with the Hamiltonian.

\begin{figure}[t]
\center
\includegraphics[width=\columnwidth]{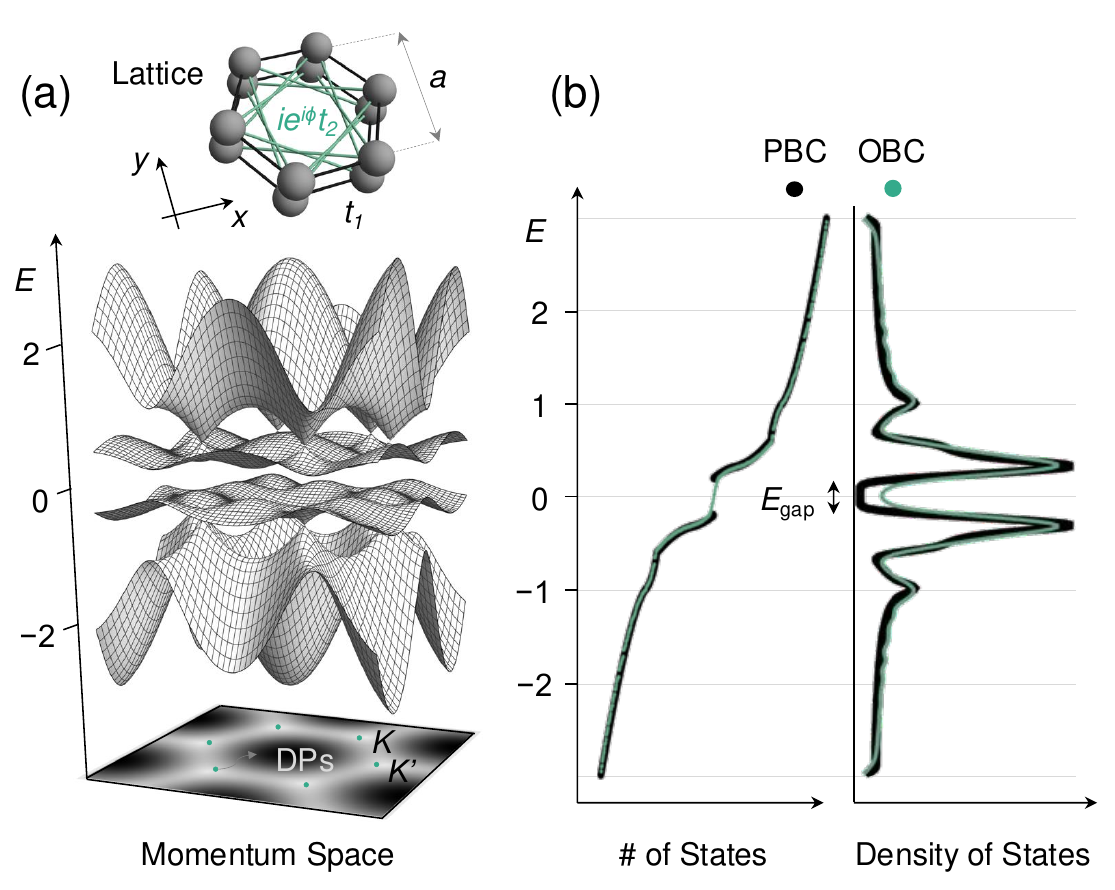}
\caption{
Twisted bilayer graphene mori{\' e} lattice model and its characteristic features.
\text{(a)} Schematic representation for the unit-cell lattice and the calculated band structure.  The band gap width is denoted by $E_{\textrm{gap}}$ and Dirac points (DPs) are shown at the $K$ and $K'$.  High-symmetry points are $\Gamma = (0, 0)$, $M = (2\pi/\sqrt{3}a, 0)$, and $K = (2\pi/\sqrt{3}a, 2\pi/3a)$.
\text{(b)} Energy eigenvalues for the number of states and its density of states under periodic (PBC) and open (OBC) boundary conditions. Here, $t$ is an overall energy scale.
}
\label{fig:figsm_01}
\end{figure}

\begin{figure}[h]
\center
\includegraphics[width=\columnwidth]{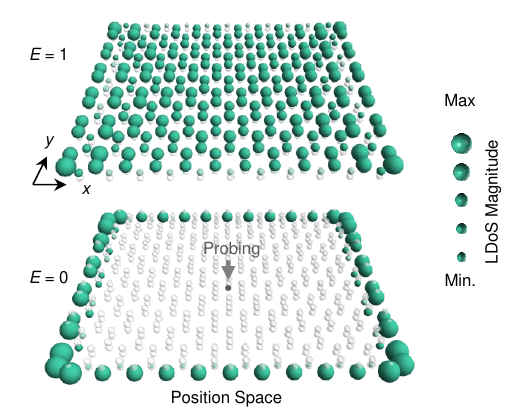}
\caption{
LDoS of the finite TBG model. LDoS within ($E=0$) and outside ($E=t$) the bulk band gap. The magnitude of the LDoS at specific energies is represented by the size of the green spheres at each lattice sites.
}
\label{fig:figsm_0}
\end{figure}

\section{The TBG lattice model, edge states, and the effects of disorder}

In this section, we present the basic descriptions of TBG lattice model and calculations of edge states arising in the open boundary TBG model. We consider a four-band tight-binding model known to describe the essential physics of the nearly flat bands of TBG according to a low-energy continuum theory. As depicted in Fig.~\ref{fig:figsm_01}\text{(a)}, the model consists of a honeycomb lattice with two layers. The explicit expression of the Hamiltonian for this lattice model in momentum space is given by
\begin{multline}
\label{eq:mo_H}
h(\vec{k}) = 
\hat{t}_1 \otimes \left[ \left( 1 + 2 \cos \frac{\sqrt{3} k_x a}{2} \cos \frac{k_y a}{2} \right) \sigma_x \right.  \\ 
\left. \quad + 2 \sin \frac{\sqrt{3} k_x a}{2} \cos \frac{k_y a}{2} \sigma_y \right]  
+ \hat{t}_2 \otimes \mathbf{1} \left[ n(\vec{k}, \phi)+n^*(\vec{k}, \phi) \right],       
\end{multline}
\noindent
with
\begin{align}
\label{eq:NNN}
\begin{split}
n(\vec{k}, \phi) = i e^{i\phi} \left(e^{i\mathbf{k}\cdot\mathbf{a}_1}+e^{-i\mathbf{k}\cdot\mathbf{a}_2}+e^{i\mathbf{k}\cdot\mathbf{a}_3}   \right) ,    
\end{split}
\end{align}
where we employ $\hat{t}_1 = 0.4 t\mathbf{1} + 0.6  t \tau_z$ and $\hat{t}_2 =  0.1  t \tau_x$, meaning intra- and inter-layer hopping amplitudes, respectively, as schematically shown by the black and green lines in Fig.~\ref{fig:figsm_01}\text{(a)}. Here, $\mathit{t}$ indicates the overall energy scale, $\mathbf{k} = (k_x, k_y)$ is the in-plane momentum, and the primitive vectors are $\mathbf{a}_{1,2} = (\sqrt{3}, \pm 1)a/2$ and $\mathbf{a}_3=\mathbf{a}_1-\mathbf{a}_2$. The Pauli matrices $\sigma_{x,y,z}$ and $\tau_{x,y,z}$ represent the sublattice and orbital degrees of freedom, respectively. The green lines, spirally connecting inter-layer sites, denote the next nearest neighbor (NNN) hoppings that are crucial for inducing a nontrivial fragile band gap. We consider the additional inter-layer hopping phase to be $\phi = 0$ here so that the initial coefficient of inter-layer hopping is purely imaginary. In Fig.~\ref{fig:figsm_01}\text{(a)}, we show the bulk band structure of the four-band model. All four energy bands are symmetric with respect to the Fermi level and each pair of bands exhibits Dirac points (DPs) at each of the $K$ and $K'$ points throughout the Brillouin zone. These DPs are protected by space-time $C_2\mathcal{T}$ inversion symmetry~\cite{Ahn2019}, where $\mathcal{T}^2 = \mathbf{1}$ and $C_2$ denotes a twofold rotation about the $z$-axis.
In addition, the valence bands exhibit an obstruction that prohibits their representation by exponentially localized Wannier functions that obey the system's $C_2\mathcal{T}$-symmetry. As this obstruction disappears when more trivial bands are added to the model, this system exhibits fragile topology~\cite{Ahn2019,Po2019}.

\begin{figure}[t]
\center
\includegraphics[width=\columnwidth]{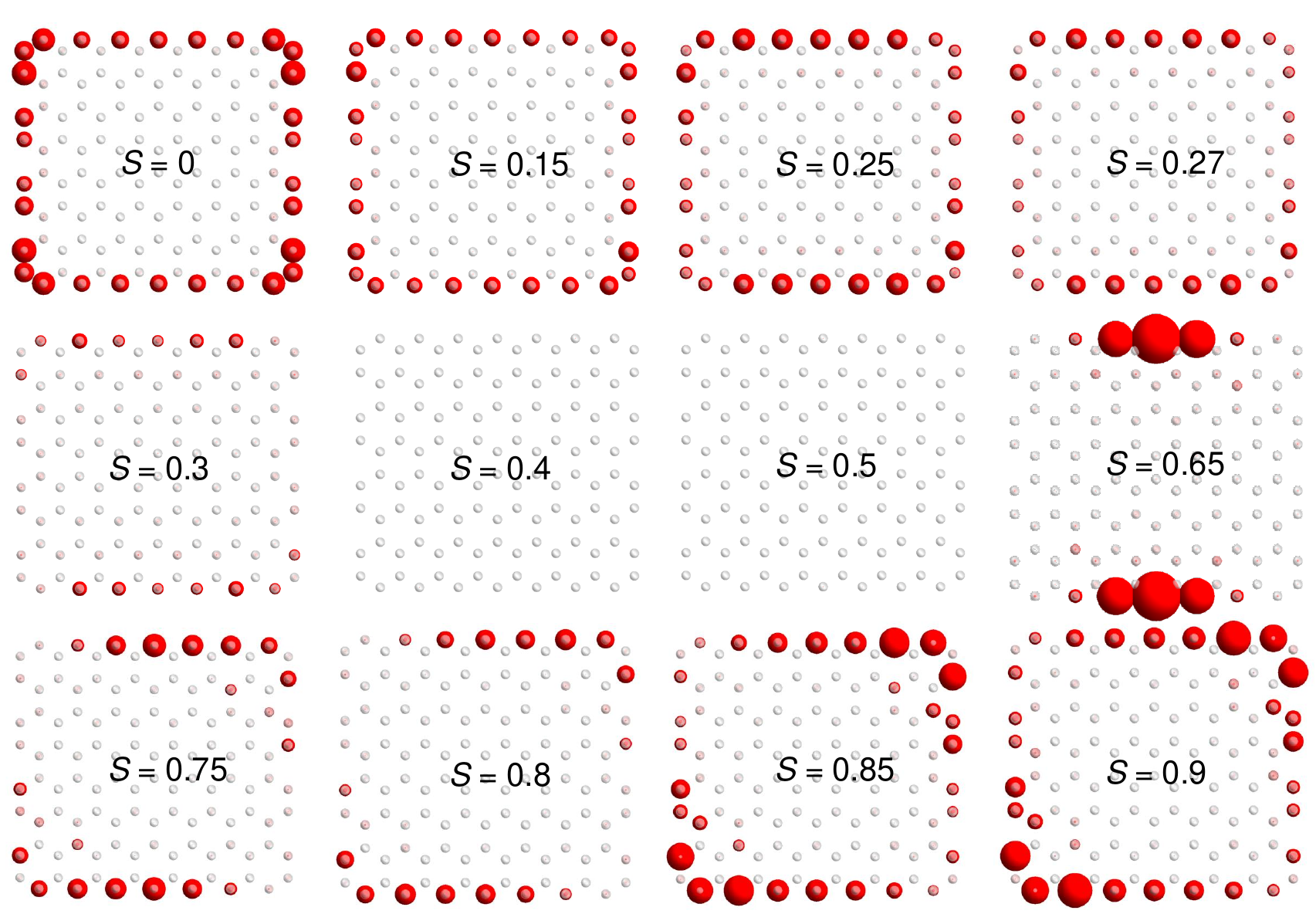}
\caption{
LDoS for the TBG model at $E$ = 0 with $C_2 \mathcal{T}$-symmetry-preserving disorder. Each plot show the LDoS for a single disorder configuration for each selected disorder strength $S$.}
\label{fig:figsm4}
\end{figure}

However, to study finite geometries and disorder in the four-band TBG model, we instead work with the system's expression in position-space, which is given by the Hamiltonian
\begin{equation}
\label{eq:po_H}
H = \sum_{\langle i,j \rangle} c_i^\dagger (\hat{t}_1)_{ij}c_j + \sum_{\langle\langle i,j \rangle\rangle} c_i^\dagger s_{ij}( ie^{i \phi}\hat{t}_2)_{ij} c_j, 
\end{equation}
where $s_{ij}=+1$ is chosen for $\mathbf{r}_i=\mathbf{r}_j+a\mathbf{y}$. We 
show the energy spectra and density of states of Eq.~\eqref{eq:po_H} under periodic (PBC) and open (OBC) boundary conditions in Fig.~\ref{fig:figsm_01}\text{(b)}, which are well matched except in the bulk band gap. Under PBC, there is a band gap corresponding to $E_{\textrm{gap}}=0.4t$, whereas under OBC, the spectrum becomes gapless due to the presence of trivial edge states ~\cite{Wu2024,Fleischmann2018,Chen2023} (see Supplementary Sec.~SII).

In Fig.~\ref{fig:figsm_0}, we show the local density of states (LDoS) calculations for two energies of the TBG model with open boundaries. The magnitude of the LDoS at each site is represented by the size of the green spheres. Unlike the LDoS at $E=t$ within the bulk bands, which is uniform across the entire system, the LDoS at $E=0$ within the band gap is localized along the system boundary. As discussed in the main text, these trivial edge states are susceptible to disorder, in contrast to those that arise from strong and stable topology such as chiral edge states in Chern insulators. Figure~\ref{fig:figsm4} shows the evolution of the LDoS at $E=0$ as the disorder strength $S$ increases (using the same definition of this parameter from the main text). Note that the localized edge states, which are uniformly distributed along the system boundary in the clean system ($S=0$), become unevenly distributed as $S$ increases. These edge states disappear as the system enters the trivial phase within a certain range of $S$, and reappear with the re-entrant fragile topological phase at strong $S$. Figure~\ref{fig:figsm5} shows the LDoS at $E=0$ for different disorder samples at $S=0.85$. As clearly shown in this result, the edge states are unevenly distributed along the boundary due to the disorder.

\begin{figure}[t ]
\center
\includegraphics[width=\columnwidth]{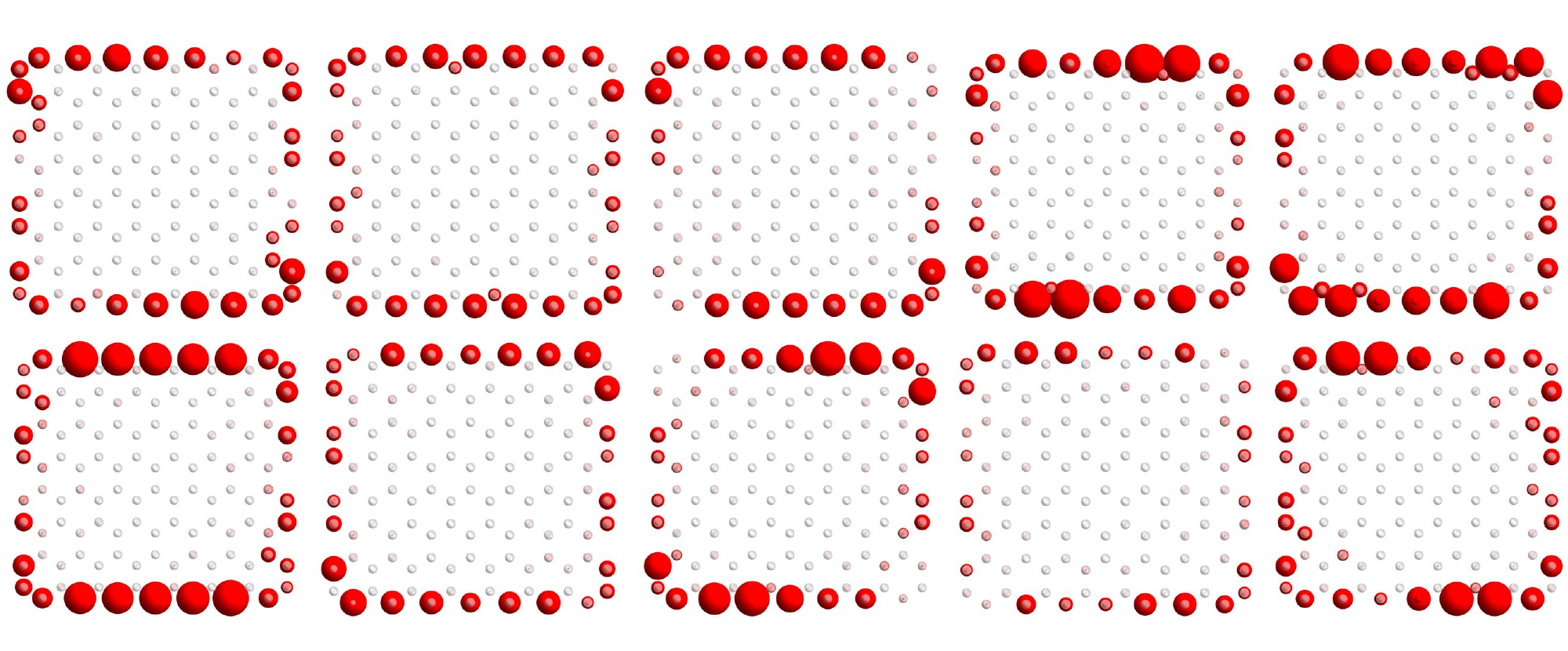}
\caption{Susceptible nature of the TBG boundary states due to the effects of disorder.
LDoS at $E$ = 0 and $S = 0.85$ for 10 different disorder samples.
}
\label{fig:figsm5}
\end{figure}

\begin{figure}[h]
\center
\includegraphics[width=\columnwidth]{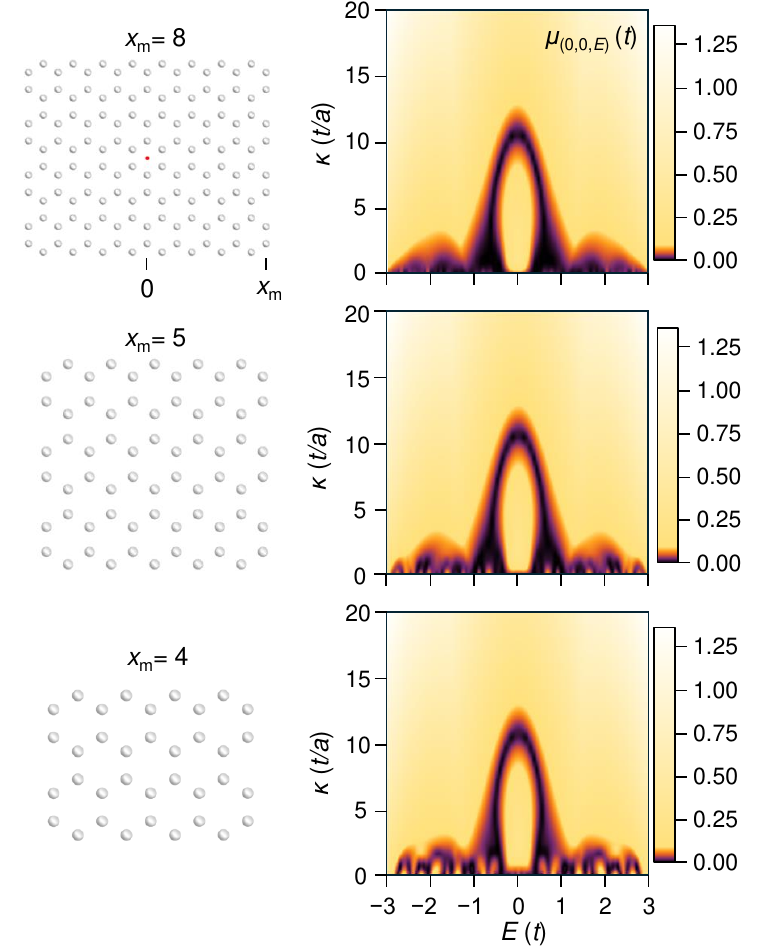}
\caption{
Consistency of the local gap in the $\kappa$-energy domain with respect to changes in total system size.
}
\label{fig:figsm1_1}
\end{figure}

\begin{figure}[t]
\center
\includegraphics[width=\columnwidth]{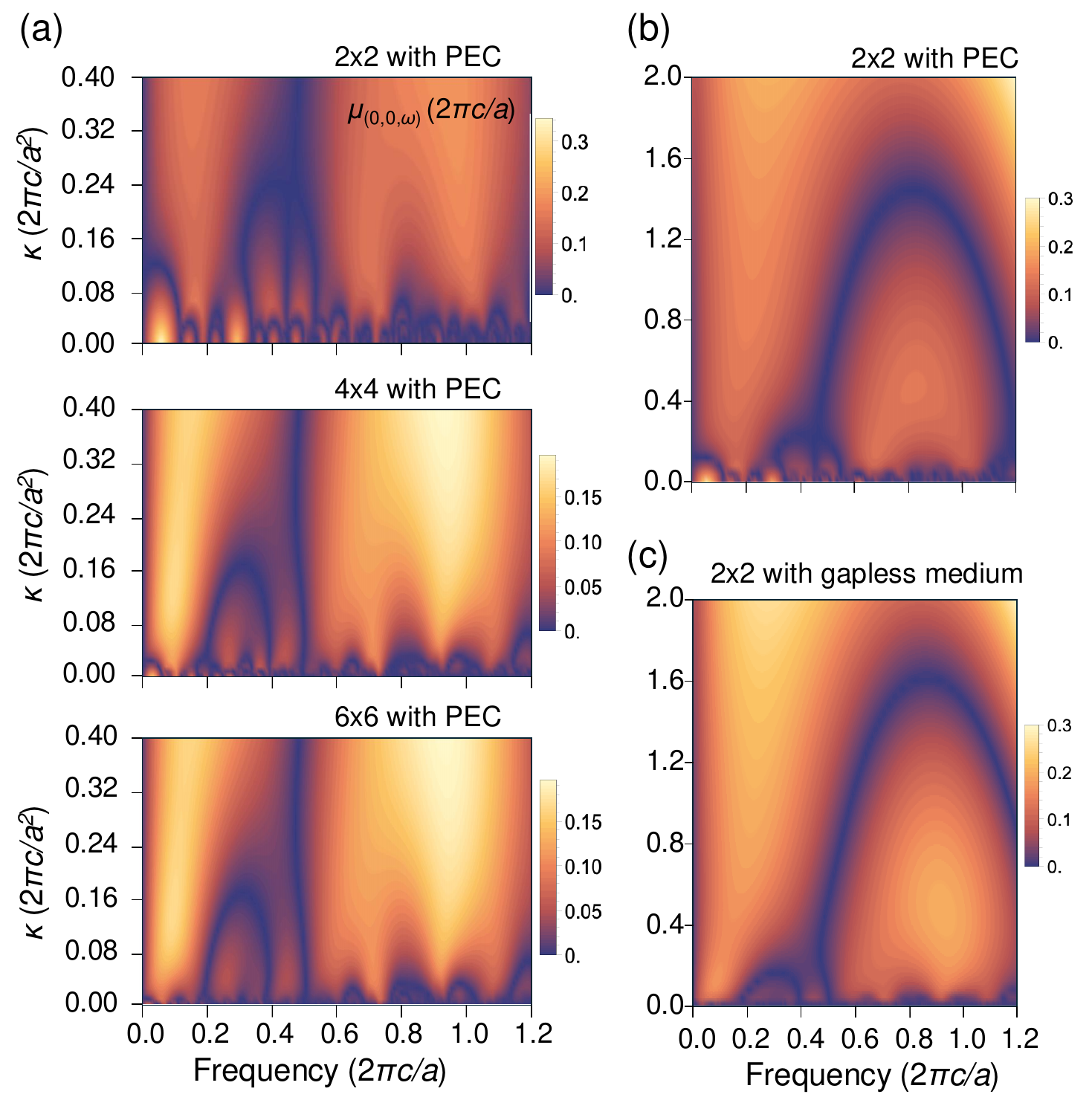}
\caption{
Consistency of the local gap behavior in the $\kappa$-frequency domain with respect to changes in total system size and the environment.
\text{(a)} Local gap $\mu$ calculation results for $2\times2$, $4\times4$, and $6\times6$ with PEC boundaries. 
\text{(b)}, \textbf{c} Large scale view of local gap $\mu$ calculation results for $2\times2$ with PEC boundary \text{(b)} and surrounded by air \text{(c)}.
}
\label{fig:figsm1_2}
\end{figure}

\section{A scaling coefficient and size dependence of a local gap}

Figures \ref{fig:figsm1_1} and \ref{fig:figsm1_2} illustrate the calculations of the local gap $\mu$ in the $\kappa$-energy domain for lattice and continuum models, respectively, indicating that these properties are consistently maintained regardless of the total system size and the surrounding environment. In Fig.~\ref{fig:figsm1_1}, the local gap is calculated within the energy range from $-3$ to $+3$ and $\kappa$ from 0 to 20$t/a$. Considering the shift in homotopy invariant upon touching $\mu = 0$, it is evident that $\mu$ robustly protects the topology near the bulk band gap at $E = 0$. This topological phase becomes trivialized after the local gap touches zero near $\kappa=12t/a$ as $\kappa$ increases. Note that the local gap behavior in the $\kappa$-$E$ domain remains consistent despite changes in the spatial size of the lattice, with $x_m=8$, $5$, and $4$. Similarly, Figure \ref{fig:figsm1_2} presents calculations of the local gap for the continuum model. Figure \ref{fig:figsm1_2}\text{(a)} shows the results for three different unit cell counts under PEC boundary conditions. Figure \ref{fig:figsm1_2}\text{(b)} demonstrates the results at larger $\kappa$ scales for the PEC and air boundaries. These results confirm the consistency of the spectral localizer framework through the large local gap values near the fragile band gap and the qualitative agreement in the $\kappa$-$\omega$ domain.

\begin{figure}[h]
\center
\includegraphics[width=\columnwidth]{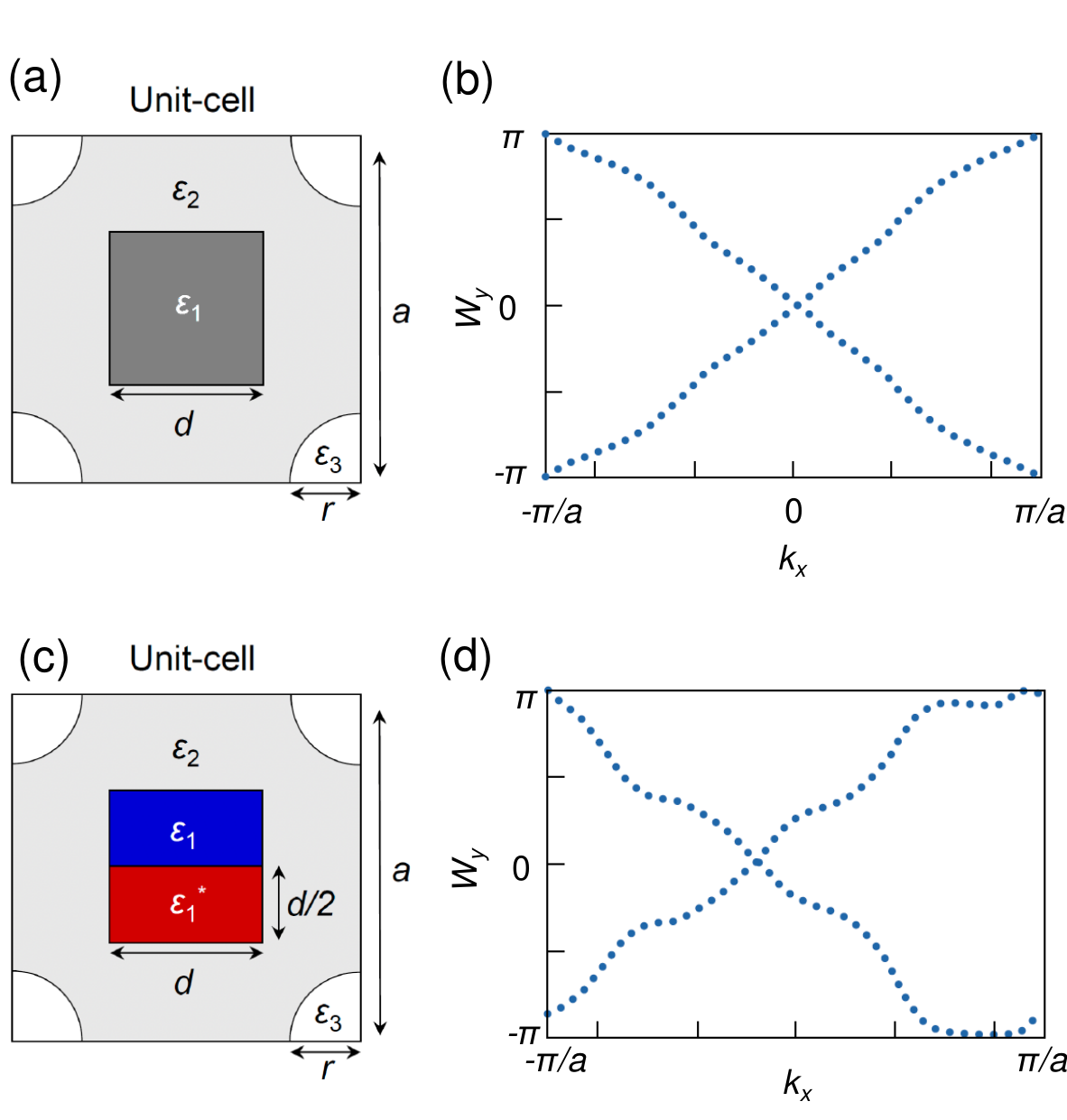}
\caption{
Wilson loop calculations for photonic crystals with fragile bands.
\text{(a)} Unit cell of the photonic crystal with both $C_2$ and $\mathcal{T}$ symmetries. Here, $\epsilon_1 = 16$, $\epsilon_2=4$ and $\epsilon_3=1$. $a$ is a period, and $r$ = 0.2$a$ and $d$ = 0.4$a$ and $G$ = $2\pi/a$.  
\text{(b)} Wilson loop eigenvalues $W_y$ for the TE-polarized bands $8+9$, plotted as a function of $k_x$. The winding of the eigenvalues indicates the non-Wannierizablility of the bands. The crossing of the eigenvalues at $k_x = 0$ and $\pi$ is enforced by $C_2$ symmetry.
\text{(c)} Unit cell of the photonic crystal with $C_2\mathcal{T}$ but with individually broken $C_2$ and $\mathcal{T}$ symmetries. Here, $\epsilon_1$ is the anisotropic permittivity tensor given in the main text, $\epsilon_2=4$ and $\epsilon_3=1$. $a$ is a period, and $r$ = 0.2$a$ and $d$ = 0.45$a$ and $G$ = $2\pi/a$.
\text{(d)} Wilson loop eigenvalues $W_y$ for the TE-polarized bands $8+9$, plotted as a function of $k_x$. Under $C_2\mathcal{T}$ symmetry, the winding in the spectrum is retained and the crossing of the eigenvalues is stable. However, the crossing can occur at arbitrary $k_x$.
}
\label{fig:figsm3}
\end{figure}

\section{Wilson loops for photonic crystals with fragile bands\label{sec:S2}}

A standard momentum-space method for characterizing band topology is to compute the Wilson loop spectrum. The Wilson loop operator over a closed loop $l$ is given by
\begin{equation}
    \mathcal{W}_{m,n} = \mathcal{P}\exp \left( i\oint_{l} \mathbf{A}_{m,n}(\mathbf{k})\cdot d\mathbf{k} \right),
\end{equation}
where $\mathbf{A}_{m,n}(\mathbf{k}) = \langle u_{m,\mathbf{k}}| i\nabla_\mathbf{k} u_{n,\mathbf{k}}\rangle$ is the Berry connection defined over the space-periodic part of the eigenstates of interest, $u_{m,\mathbf{k}}(\mathbf{r})$. The Wilson loop operator is a matrix that encodes the geometric phases traced by eigenstates along closed loops in momentum space. Specifically, the eigenvalues of $\mathcal{W}$ represent the Berry phases accumulated along the loop. A non-trivial winding in the Wilson loop eigenvalue spectrum signifies the non-Wannierizability of the bands, indicating their non-trivial topology. In the case of photonic crystals (PhCs), this calculation can be performed using the electromagnetic eigenmodes extracted from finite-element or plane-wave expansion methods~\cite{de2019tutorial, DePaz2019, Vaidya2023}.

Here, we calculate the Wilson loop spectra for the photonic crystal (PhC) structures studied in the main text. We compute the eigenvalues of $\mathcal{W}$ along non-contractible paths in momentum space with a fixed $k_y$ and plot them as a function of $k_x$. First, let us consider the PhC with the unit cell shown in Fig.~\ref{fig:figsm3}\text{(a)} with both $C_2$ and $\mathcal{T}$ symmetries. The Wilson loop spectrum for TE-polarized bands 8 and 9 is shown in Fig.~\ref{fig:figsm3}\text{(b)}. We observe a double winding of the eigenvalues in the spectrum indicating that these bands are fragile. Moreover, the crossings at $k_x = 0$ and $\pi$ are enforced by the presence of $C_2$ symmetry. In Fig.~\ref{fig:figsm3}\text{(c)} we consider the PhC studied in the main text with broken $C_2$ and $\mathcal{T}$ symmetries but with preserved $C_2\mathcal{T}$. The corresponding Wilson loop spectrum shown in Fig.~\ref{fig:figsm3}\text{(d)} continues to exhibits a double winding of the eigenvalues indicating fragile topology. However, under $C_2\mathcal{T}$, these crossings are generic and can occur at arbitrary $k_x$~\cite{Bradlyn2019, Song2019, Bouhon2019}. 

We emphasize that this momentum-space method is only valid for infinite, spatially periodic structures and fails in the presence of boundaries, gapless environments, or disorder. Furthermore, position-space symmetry indicators, which can identify fragile topology, require spatial symmetries that are absent in the structures considered above~\cite{Song2020}. In contrast, the position-space methods introduced in the main text are applicable to gapless and finite heterostructures, such as those shown in Fig.~4\text{(d)}, as well as to systems with disorder, as illustrated in Fig.~3 of the main text.

\bibliography{ref}

\end{document}